\documentclass[review]{elsarticle}

\journal{International Journal of Approximate Reasoning}
\usepackage{amsmath} %,dsfont}
\usepackage{graphicx}
\usepackage{color}
\usepackage{tikz}
\usepackage{indentfirst}

\usepackage{subfigure}
\usepackage{graphicx}
\usepackage{epsfig}
\usepackage{amssymb,times} %,dsfont}

\usepackage{xspace,pgf,latexsym}
\usepackage{lineno}
\usepackage{color}
\usepackage{verbatim}
\usepackage{enumitem}
\usepackage{MnSymbol,stmaryrd}

\newtheorem{theorem}{Theorem}[section]
\newtheorem{proposition}{Proposition}[section]
\newtheorem{corollary}{Corollary}[section]
\newtheorem{lemma}{Lemma}[section]
\newtheorem{definition}{Definition}[section]
\newtheorem{remark}{Remark}[section]
\newtheorem{example}{Example}[section]

\newcommand{\prt}[1]{\langle #1\rangle}

\newenvironment{proof}{\noindent\textbf{Proof: }}{\hfill \small $\Box$}

\newcommand{\igdef}{\stackrel{def}{=}}
\newcommand{\raw}{\rightarrow}

\newcommand{\Twoheadrightarrow}{{\Rightarrow \!\! >} }

\newcommand{\Rawint}{\Rightarrow \!\! >}

\newcommand{\RR}{\mathbb{R} }

\newcommand{\myllle}{\prec\hspace{-0.7ex}\prec}

\begin{document}

\begin{frontmatter}

\title{Semi-BCI Algebras}
%\tnotetext[mytitlenote]{Fully documented templates are available in the elsarticle package on %\href{http://www.ctan.org/tex-archive/macros/latex/contrib/elsarticle}{CTAN}.
%}

\author{Regivan H.N. Santiago$^a$, Benjamin Bedregal$^a$, Jo\~ao Marcos$^a$, Carlos Caleiro$^b$, Jocivania Pinheiro$^{a,c}$}
%\address{Universidade Federal Rural do Semi-Árido, \\
%Mossor\'{o} - Brazil}


\address[addA]{Group for Logic, Language, Information, 
	Theory and Applications -- \textbf{LoLITA}\\
	Department of Informatics and Applied Mathematics -- \textbf{DIMAp}\\
	Federal University of Rio Grande do Norte --  \textbf{UFRN}\\
	59.072-970. Natal, RN, Brazil} 
\address[addB]{Instituto de Telecomunicações\\
	Department of Mathematics --- IST, Universidade de Lisboa\\
	1049-001. Lisboa, Portugal}
\address[addC]{Group for Theory of Computation, Logic and Fuzzy Mathematics\\
	Department of Natural Sciences, Mathematics and Statistics --- \textbf{DCME}\\
	Center of Exact and Natural Science --- \textbf{CCEN}\\
	Rural Federal University of SemiArid --  \textbf{UFERSA}\\
	59.625-900. Mossoró, RN, Brazil}

\begin{abstract}
	The notion of semi-BCI algebras is introduced and some of its properties are investigated. This algebra is another generalization for BCI-algebras. It  arises from the ``intervalization'' of BCI algebras. Semi-BCI  have a similar structure  to  Pseudo-BCI algebras however they are  not the same. In this paper we also provide an investigation on the similarity between these classes of algebras by showing how they relate to the process of intervalization.
\end{abstract}


\end{frontmatter}

\section{Introduction}

%In 1998, Petr Hájek published a book on the foundation of Fuzzy Logics \cite{Hajek1998}. He proposed an algebraic structure to be the model for such logics --- BL-algebras.  A BL-algebra is an MV-algebra which is a BCI algebra with extra properties \cite{Dvurecenskij2000}. 

One of the most well known references on the algebraic approach to logics is the book of Rasiowa \cite{Rasiowa1974} which dates to 70s. In this book, at 
pages 16-17,  the notion of implicative algebra, which aim at modelling a simple notion of implication is provided: An \textbf{implicative algebra} is an algebra $\langle A,\Rightarrow,\top \rangle$ of type $(2,0)$ which satisfies the following properties:

\begin{enumerate}[labelindent=\parindent, leftmargin=*,label=(i-\arabic*)]
	\item $a\Rightarrow a=\top$,
	\item If $a\Rightarrow b=\top$ and $b\Rightarrow c=\top$, then $a\Rightarrow c=\top$,
	\item  If $a\Rightarrow b=\top$ and $b\Rightarrow a=\top$, then $a=b$.
\end{enumerate}

%\textcolor{red}{This paper proposes a new algebraic structure which generalizes the notion of BCI-algebra. It is an algebraic structure which captures the most important properties of a Fuzzy Implication after it has been ``intervalized'' in a correct way. The resulting structures, called \textbf{Semi-BCI algebras}, capture the properties of  the structures which arise from the ``intervalization'' of BCI-algebras. In other words, as BCI-algebras abstract the \L ukasiewicz algebra, the SBCI-algebras abstract the respective intervalization of the \L ukasiewicz algebra. The paper also provides the connection of such structures with PBCI-algebras.}

A direct consequence of such definition is the establishment of an order relation ``$\leq$'' on $A$, which is known as \textbf{Order Property (OP)} of implications:

\begin{equation}
a\leq b\mbox{ if and only if } a\Rightarrow b=\top.
\end{equation}

As a consequence, many implications are axiomatized preserving \textbf{(OP)}. 
However, some interesting implications in the field of fuzzy logics do not satisfy such requirement, for example (see \cite{Bac2008Book,PBSS18,Yager1980}):

%\begin{enumerate}[labelindent=\parindent, leftmargin=*,label=]

%	\item 
Consider the  algebra $\langle [0,1],\rightarrow_{\mbox{\tiny YG}},1 \rangle$, such that:

\begin{equation*}
x \rightarrow_{\mbox{\tiny YG}} y = 
\begin{cases}
1 \text{, if $x=y=0$}
\\
y^x \text{, otherwise}
\end{cases}.
\end{equation*}

In this case $0.3\leq 0.5$, but $0.3\rightarrow_{YG} 0.5\approx 0.81225$. However, ``$x\rightarrow_{YG} y=1$ implies $x\leq y$''.
%\end{enumerate}

The interval counterpart of \L ukasiewicz implication introduced by Bedregal and Santiago \cite{BS2013a} also fails to satisfy \textbf{(OP)}. The authors, however, revealed that the resulting implication satisfy:
%\begin{quotation}\hfill
\begin{enumerate}
	\item if $X\ll Y$\footnote{$X\ll Y$ iff $\overline{X}<\underline{Y}$.}, then $X\rightarrow Y=1$;
	\item if $X\rightarrow Y=1$, then $X\leq_{KM} Y$.
\end{enumerate}
%\end{quotation}

The relation ``$\ll$'' is precisely the way-below relation \cite{Gierz2003} of the usual Kulisch-Miranker order on intervals ``$\leq_{KM}$''. Way-below relations, ``$\myllle $'', are auxiliary relations \cite{Gierz2003} of partial orders ``$\preccurlyeq$''; they have the 
following properties:

\begin{enumerate}
	\item if $x\myllle y$, then $x\preccurlyeq y$.
	\item if $u\preccurlyeq x\myllle  y\preccurlyeq z$, then $u\myllle  z$.
	\item if a smallest element $0$ exists, then $0\myllle  x$.
\end{enumerate} 

Since \textbf{(OP)} connects an implication to the underlying order relation and this is connected to auxiliary relations, this paper proposes to internalize such connection through two implications; one connected to the usual partial order and the other connected to its way-below relation. The resulting algebraic structure is called \textbf{semi-BCI algebra} which abstracts both BCI-algebras and their intervalization.

Another generalization for BCI-algebras which contains two implications is called \textbf{Pseudo-BCI algebra} which was proposed by W. A. Dudek and Y. B. Jun \cite{Dudek}. The connection of such algebras to Semi-BCIs is investigated here.	  

Since SBCIs encompass both BCIs and its interval counterpart they tend to model logics in which the notion of impreciseness is required. 

The paper is organized in the following way: Section \ref{Sec-BCI} provides a brief review of BCI-algebras and their properties. 
Section \ref{SecIntEst} provides an overview of the intervalization process. Section \ref{Sec-ISBCIA} shows the intervalization of BCI-algebras and some properties of this interval algebra. Section \ref{Sec-SBCIA} introduces the notion of semi-BCI algebra and prove some of its properties. Section \ref{SecPseudo} discusses the relation between Semi-BCI algebras and Pseudo-BCI algebras. Finally, section \ref{SecFinalRem} provides some concluding remarks.

\section{BCI-Algebras}\label{Sec-BCI}

BCI-algebras are mathematical structures for modelling fuzzy logics. They were introduced by Iséki \cite{Iseki1966} in the 60's and since then have been extensively investigated. There are several axiom systems for BCI-algebras. We will present here the axiom systems defined by \cite{Hua06}, in which he assures that the BCIs are algebras of the form $\langle A,\ast,\perp \rangle$ which satisfy the following properties:

\begin{enumerate}[labelindent=\parindent, leftmargin=*,label=BCI-\arabic*]
	\item $((x \ast y)\ast(x\ast z))\ast (z\ast y)=\perp$,
	\item $(x\ast(x\ast y))\ast y=\perp$,
	\item $x\ast x=\perp$,
	\item $x\ast y=\perp$ and $y\ast x=\perp\ \Longrightarrow x=y$.
\end{enumerate}

A BCI-algebra is called BCK-algebra if it also satisfies:

\begin{enumerate}[labelindent=\parindent, leftmargin=*,label=BCK-\arabic*]
	\item $\perp \ast x=\perp$
\end{enumerate}

On any BCI-algebra it is possible to define a partial order ``$\unlhd$'' as: ``$x\unlhd y$ iff $x\ast y=\perp$''. 
Therefore, a BCI-algebra is BCK if and only if $\perp$ is its least element.

\begin{example} The following algebras are BCI.
	
	\begin{enumerate}
		\item $\langle [0,+\infty),\ast,0 \rangle$, s.t. $x \ast y=\max\{0,x-y\}$.
		\item $\langle \mathcal{P}(X),\ominus,\emptyset \rangle$, where $A\ominus B$ is the set difference between $A$ and $B$.
	\end{enumerate}
\end{example}

BCI-logics interpret the Curry combinators: (B) $\lambda xyz.x(yz)$, (C) $\lambda xyz.xzy$ and 
(I) $\lambda x.x$ --- see \cite{hindley1986}. This set of combinators are functional counterparts for some Fuzzy Implications. 
They can also be interpreted by algebras: $\mathcal{C}=\langle A,\rightarrow,\top \rangle$ which satisfy:

\begin{enumerate}[labelindent=\parindent, leftmargin=*,label=($\mathcal{C}$-\arabic*)]
	\item $(y\rightarrow z)\rightarrow ((z\rightarrow x)\rightarrow(y\rightarrow x))=\top$,\label{PropC1}
	\item $x\rightarrow ((x\rightarrow y)\rightarrow y)=\top$,\label{PropC2}
	\item $x\rightarrow x=\top$, \label{Curry-I}
	\item $x\rightarrow y=\top$ and $y\rightarrow x=\top$ imply $x=y$.\label{PropC4}
\end{enumerate}

On any such structure it is possible to define a partial order ``$\preceq$'' as: 

\begin{enumerate}[labelindent=\parindent, leftmargin=*,label=($\mathcal{C}$-\arabic*)]
	\setcounter{enumi}{4}
	\item $x\preceq y$ iff $x\rightarrow y=\top$.\label{PropC5}
\end{enumerate} 

\begin{example} The following algebra satisfies properties \normalfont{\ref{PropC1}-\ref{PropC5}}:
	\begin{enumerate}
		\item $\langle [0,1],\rightarrow,1 \rangle$, s.t. $x\rightarrow y=\min (1,1-x+y)$.
	\end{enumerate}
\end{example}

There is a way to obtain the above axioms from those of BCI-algebras and vice-versa, the correspondence can be obtained in the following way:

\begin{proposition}\label{PropEqvBCIdBCI}
	Let $\langle A, \ast,\perp \rangle$ be a BCI-algebra. The algebra $\mathcal{C}=\langle A,\rightarrow,\top \rangle$, where: $x\rightarrow y \stackrel{def}{=} y \ast x$ and 
	$\top \stackrel{def}{=} \perp$ satisfies the axioms \normalfont{\ref{PropC1}-\ref{PropC4}}.
\end{proposition}

\begin{proof}
	$(y\rightarrow z)\rightarrow ((z\rightarrow x)\rightarrow(y\rightarrow x))\igdef ((z\rightarrow x)\rightarrow(y\rightarrow x)) \ast 
	(y\rightarrow z)\igdef\dots\igdef ((x\ast y)\ast(x\ast z))\ast (z\ast y)=\perp$ (BCI-1). 
	But, $\top\igdef\perp$. The other axioms are similarly proved.	
\end{proof}

\begin{proposition}
	Let $\mathcal{C}=\langle A,\rightarrow,\top \rangle$ be an algebra satisfying properties \normalfont{\ref{PropC1}-\ref{PropC4}}. The algebra 
	$\langle A,\ast,\perp \rangle$, where: $x\ast y\igdef y\rightarrow x$ and $\perp\igdef\top$ is a BCI-algebra. 
\end{proposition}

\begin{proof}
	Analogous to Proposition \ref{PropEqvBCIdBCI}.
	
\end{proof}

\begin{corollary}
	The relation ``$\preceq$'' is the dual partial order of ``$\unlhd$''; namely $x\preceq y$ if and only if $y\unlhd x$. 
\end{corollary}

\paragraph{Terminology } Since one kind of algebra can be obtained from the other and any result obtained for one  can be easily translated, by duality, to the other, both structures are called
BCI-algebras. This work consider the second kind of structure for our generalization. In this context, whenever $x\rightarrow\top=\top$, i.e. $x\preceq\top$, the BCI-algebra $\mathcal{C}=\langle A,\rightarrow,\top \rangle$ will be called \textbf{BCK-algebra}.   Now, let be some  properties of BCI-algebras:

\paragraph{Some Properties of BCI-Algebra } (for more details see \cite{Hua06})

\renewcommand{\thefootnote}{\Roman{footnote}}

\begin{enumerate}[labelindent=\parindent, leftmargin=*,label=(A-\arabic*)]
	\item $\top\preceq x$ implies $x=\top$,\label{BCI-Pr1}
	\item $x\preceq y$ implies $y\rightarrow z\preceq x\rightarrow z$,\label{BCI-Pr2} --- (First place antitonicity)
	\item $x\preceq y$ implies $z\rightarrow x\preceq z\rightarrow y$,\label{BCI-Pr7}  --- (Second place isotonicity) 	
	\item $x\preceq y$ and $y\preceq z$ implies $x\preceq z$,\label{BCI-Trans}
	\item $x\rightarrow (y\rightarrow z)=y\rightarrow (x\rightarrow z)$,\label{Curry-C} --- (Exchange)
	\item $x\preceq y\rightarrow z$ implies $y\preceq x\rightarrow z$,\label{BCI-Pr5}
	\item $x\rightarrow y\preceq (z\rightarrow x)\rightarrow (z\rightarrow y)$,\label{Curry-B}\label{BCI-Pr6}
	\item $\top\rightarrow x=x$,\label{BCI-Pr8} --- (Left Neutrality)
	\item $((y\rightarrow x)\rightarrow x)\rightarrow x=y\rightarrow x$,\label{BCI-Pr9}
	\item $x\rightarrow y\preceq (y\rightarrow x)\rightarrow\top$,
	\item $(x\rightarrow y)\rightarrow\top=(x\rightarrow\top)\rightarrow (y\rightarrow\top)$. \label{BCI-Pr11}	
\end{enumerate}

Properties \ref{Curry-B}, \ref{Curry-C} and \ref{Curry-I} model the combinators B, C and I of Combinatorial Logic \cite{hindley1986}.

\begin{proposition} \label{pro-BCI-BCK}
	Let $\prt{A,\rightarrow,\top}$ be a BCI-algebra. $\prt{A,\rightarrow,\top}$ is a BCK-algebra if and only if  for each $x\in A$ there exists $y\in A$ 
	such that $y\preceq x$ and $y\preceq \top$.  
\end{proposition}

\begin{proof} ($\Rightarrow$) Straightforward because in BCK-algebras $\top$ is the greatest element, i.e. $x\preceq \top$ for each $x\in A$.
	
	($\Leftarrow$)
	Suppose that $\prt{A,\rightarrow,\top}$ is not a BCK-algebra. Then, there exists $a\in A$ such that $a\not \preceq \top$. By hypothesis there exists 
	$b\in A$ such that $b\preceq a$ and $b\preceq \top$. So, by \ref{BCI-Pr8} and definition of $\preceq$, 
	$(b\rightarrow \top)\rightarrow ((\top\rightarrow a)\rightarrow (b\rightarrow a)) = 
	\top\rightarrow (a\rightarrow \top)=a\rightarrow \top\neq \top$. Therefore, \ref{PropC1} fails.
	
\end{proof}

As stated  in the Introduction, this paper shows that the behavior of the process of intervalization does not preserve \textbf{(OP)}. In order
to precisely define what does it mean,  the next section  introduces the concept of intervalization over abstract partial orders.

\section{Intervalization of Structures}\label{SecIntEst} 

The limited capacity of machines to store just a finite set of finitely represented objects constraints the automatic calculation (computation) of structures in which a
machine representation of some objects exceeds such capacity. In the case of real numbers, although most programs provide highly 
accurate results, it can happen that rounding errors built up during each step in the computation produce results which are not even meaningful. 
For more details see the early Forsythe’s report \cite{Forsythe1970}. In 1988, Siegfried Rump \cite{Rump1988} published the result of a computed 
function in an IBM S/370 mainframe. The function was:

\begin{equation}
y=333.75b^6+a^2(11a^2b^2-b^6-121b^4-2)+5.5b^8+\frac{\displaystyle a}{\displaystyle 2b}.\label{EqRump}
\end{equation}

He calculated for $a=77617.0$ and $b=33096.0$, and the result was:

\begin{enumerate}[labelindent=\parindent, leftmargin=*,label=\arabic*.]
	\item single precision: $y=1.172603\dots$;
	\item double precision: $y=1.1726039400531\dots$;
	\item extended precision: $y=1.172603940053178\dots$.
\end{enumerate}

All results lead any user to conclude that IBM S/370 returned the
correct result. However this result is WRONG and the correct result lies in the interval: $-0.82739605994682135\pm 5\times 10^{-17}$. Note that \textit{even the sign is wrong!}

One of the proposals to overcome this problem is due, almost simultaneously, to Ramon Moore \cite{Moore1959,MoorePhD} and Teruo Sunaga \cite{Sunaga1958}. They developed the so-called \textbf{interval arithmetic}. Interval arithmetic is a set of operations on the set of all closed intervals $\mathbb{I}(\mathbb{R})=\{[a,b]:a,b\in\mathbb{R} \text{ and } a\leq b \}$. The operations are defined in the following way:

\begin{enumerate}[labelindent=\parindent, leftmargin=*,label=\arabic*.]
	\item $[a,b]+[c,d]=[a+c,b+d]$,
	\item $[a,b]\cdot [c,d]=[\min P,\max P]$ --- where $P=\{a\cdot c,a\cdot d,b\cdot c,b\cdot d\}$,
	\item $[a,b]-[c,d]=[a-d,b-c]$,
	\item $[a,b]/[c,d]=[a,b]\cdot ([1/ b,1/ a])$; provided that $0\notin [c,d]$.
\end{enumerate}

Observe that for each operation $\diamond\in\{+,-,\cdot,/\}$, $[a,b]\diamond [c,d]=\{x\diamond y\in\RR:x\in [a,b]\wedge y\in [c,d]\}$. 
This reveals two important properties of this arithmetic (a) \textbf{Correctness} and (b) \textbf{Optimality}. 

\begin{quotation}
	``\textit{Correctness}. The criterion for correctness of a definition of interval arithmetic is that the 
	“Fundamental Theorem of Interval Arithmetic” holds \footnote{Moore \cite[Theorem 3.1, p. 21]{Moo79}: 
		If $F$ is an inclusion monotonic interval extension of $f$, then $\stackrel{\raw}{f}(X_1, . . . , X_n)\subseteq F(X_1, . . . , X_n)$; 
		where $\stackrel{\raw}{f}(X_1, . . . , Xn)=\{f(x_1,\dots,x_n):x_i\in X_i\}$.}: when an expression is evaluated using intervals, it yields 
	an interval containing all results of pointwise evaluations based on point values that are elements of the argument intervals. 
	
	[\dots]
	
	\textit{Optimality}. By optimality, we mean that the computed floating-point interval is not wider than necessary.'' 
	
	\hfill Hickey \textit{et.al}\cite[p.1040]{Hic01}
\end{quotation}

The philosophy behind intervals is the following: Enclosure in intervals the values which are not exact by any reason (e.g. the value comes from  an imprecise measurement) and apply correct and optimal operations on such intervals in order to obtain the best interval which contains the desired output.  This approach will avoid what happened with the Rump's example. Therefore, the notion of correctness is indispensable for such philosophy.

The property of correctness was investigated in 2006 by Santiago \textit{et al} \cite{BSA2013,SBA06}. In those papers, instead of correctness the authors 
used the term \textbf{representation}, since an interval computation could be understood not just as a machine representation of real numbers, but also as a 
mathematical representation of real numbers (this idea is confirmed by the Representation Theorems of Euclidean continuous functions in \cite{BSA2013,SBA06}). In what follows this notion is shown for binary operations: A binary interval operation, $\diamond$, \textit{represents} 
a binary real operation, $\diamond$, whenever:

\begin{equation}
(x,y)\in [a,b]\times [c,d]\Rightarrow x\diamond y\in [a,b]\diamond [c,d]
\end{equation}

This can be easily extended to $n$-ary operations. The authors showed that this notion is more general than what is stated by the  Fundamental Theorem of 
Interval Arithmetic; given that there are representations which are not inclusion monotonic (see \cite[p. 238]{SBA06}).

One noteworthy point which will be taken into account in the present paper: There is a difference between the  representation of a function $f$ as an interval 
function $F$ and an \textbf{extension} of a function $f$ to an interval function $G$. For example, given intervals $X=[\underline{X},\overline{X}]$ 
and $Y=[\underline{Y},\overline{Y}]$, the function 
$X-Y=[\min(\underline{X}-\underline{Y},\overline{X}-\overline{Y}),\max(\underline{X}-\underline{Y},\overline{X}-\overline{Y})]$, presented in 
\cite{Markov77}, extends the subtraction on real numbers, however $[2,3]-[2,3]=[0,0]$, $2.5,2.1\in [2,3]$, but $2.5-2.1\not\in [2,3]-[2,3]$; in 
other words, this operation is not correct. So, there are interval extensions which are not correct. \textit{They are useless for the proposed philosophy}.

The process of giving the correct and optimal interval version $F$ for a function $f$ is called: ``\textbf{intervalization}''.  There are many proposal of  intervalization of algebraic structures further than that of real numbers proposed by Moore and Sunaga. In the literature, the reader can find proposals even for the field of Logic, since there are structures which interpret logics that are susceptible to the same situation of $\mathbb{I}(\mathbb{R})$. For example: The \L ukasiewicz implicative algebra  $\prt{[0,1],\raw_{\tiny LK},0,1}$ s.t. $x\rightarrow_{\tiny LK} y=\min (1,1-x+y)$ interprets some many-valued logics and was ``\textbf{intervalized}'' by Bedregal and Santiago in  \cite{BS2013a}. Its MV-algebra counterpart was intervalized by Cabrer et al in \cite{Cabrer2014}, also,  in order to overcome the same problems already stated for $\mathbb{I}(\mathbb{R})$. In both cases, the interval algebras did not satisfy the same properties that are satisfied by the algebras that they came from. The same happened with $\mathbb{I}(\mathbb{R})$!

% Someone could state:  ``In the case of the implementation of fuzzy logics it is enough to take a finite set of values, since a human being can 
% assign only a finite amount of membership values and since some interesting implications are closed operations it does justify the application of 
% interval methods''. However, there are some fuzzy implications which are not closed operations over some finite sets; e.g. the Yager 
% implication \cite{Yager1980}:
% 
% 	  \begin{equation*}
% 	    x \rightarrow_{\mbox{\tiny YG}} y = 
% 	    \begin{cases}
% 	          1 \text{, if $x=y=0$}
% 	          \\
% 	          y^x \text{, otherwise}.
% 	    \end{cases}
% 	   \end{equation*}
% 
% On the other hand, consider the algebra, $\prt{A,\rightarrow}$, in which the elements of $A$ come from physical measurements and are real or complex 
% numbers. The error of measurement as well as those errors which come from  a sequence of $\rightarrow$ applications must be controlled, otherwise we fall in 
% the same situation of Rump's equation (\ref{EqRump}). Like the arithmetic 
% for real numbers, the algebra $\prt{A,\rightarrow}$ could be represented by an interval algebra.  In the case of \L ukasiewicz logic, 
% Cabrer and Mundici \cite{Cabrer2014} provide the notion of Interval MV algebras.

The following section a way of ``\textbf{intervalizing'}' BCI-algebras is provided. Like the case of MVs and \L ukasiewicz algebras the resulting structure does not belong to the same category of its starting algebra. This paper we provide an investigation of the resulting structures. In order to achieve that,  some required concepts, like the abstract notion of \textbf{intervals} are introduced. The aim, again, is to provide the ability to use intervals to represent the elements of  an algebra $\prt{A,\rightarrow}$.

\begin{definition}[Abstract Intervals]\label{DefAbsInterv}
	Given a poset $\prt{A,\leq}$, the set  $[a,b]=\{x\in A:a\leq x\leq b\}$ is called \textbf{the closed interal with endpoints $a$ and $b$} and 	$\mathbb{A}=\{[a,b]\in A\times A: a\leq b\}$ is the \textbf{set of all intervals of elements of $A$}. For any, $X\in \mathbb{A}$ its  
	left and right endpoints by $\underline{X}$ and $\overline{X}$, respectively, i.e. if $X=[a,b]$ then $\underline{X}=a$ and 
	$\overline{X}=b$. When $\underline{X}=\overline{X}$ the interval is called \textbf{degenerate}. The embedding $i:A\rightarrow\mathbb{A}$, s.t. $i(a)=[a,a]$ is called \textbf{natural embedding}. On the set $\mathbb{A}$ it is 
	canonical to define the partial order: $X\leq_{km} Y$ if and only if $\underline{X}\leq\underline{Y}$ and $\overline{X}\leq\overline{Y}$. This 
	relation is called \textbf{pointwise} or \textbf{Kulisch-Miranker order}.
\end{definition}

Since BCI-algebras are partially ordered systems, $\prt{B,\leq}$, it is possible to apply Definition \ref{DefAbsInterv} to  obtain the partial order $\prt{\mathbb{B},\leq_{km}}$.  The question is about the implications on $\mathbb{B}$: If an interval operation on $\mathbb{B}$ satisfies the BCI axioms, is it also correct? The following section will show that the answer is negative. But what does it mean? It means it is not possible to have both: (1) correctness and (2) the known theory of BCI-algebras for $\prt{\mathbb{B},\leq_{km}}$. So, since  correctness is indispensable, a price must be paid: \textit{A new  theory for $\prt{\mathbb{B},\leq_{km}}$} must be developed. This is the reason of this paper!

\section{Intervalization of BCI-algebras}\label{Sec-ISBCIA}

This section shows that it is not possible to have an interval BCI-algebra with a correct implication. Proposition \ref{PropInervalBCI} shows that  it is possible to build an interval BCI-algebra, but with a non-correct implication, and Theorem \ref{TheoremIntervalBCInotRep} shows that it is an impossible task. Finally, we provide the ``BCI-algebra intervalization theorem'' and some properties of resulting algebra.

\begin{lemma}\label{lem-BCI-ord}
	Let $\prt{A,\rightarrow,\top}$ be a BCI-algebra such that $\prt{A,\preceq}$ is a meet-semilattice $\prt{A,\wedge}$. For each 
	$a,b,c\in A$, $a\rightarrow (b\wedge c)=\top$ iff $a\rightarrow b=\top$ and $a\rightarrow c=\top$. Moreover, if 
	$a\rightarrow c=\top$ and $b\rightarrow c=\top$ then $(a\wedge b)\rightarrow c=\top$.
\end{lemma}

\begin{proof}
	Straightforward.
	
\end{proof}

\begin{lemma}\label{lem-BCI-ord2}
	Let $\prt{A,\rightarrow,\top}$ be  a BCI-algebra such that $\prt{A,\preceq}$ is a  meet-semilattice satisfying: 
	\begin{equation}\label{eq-(*)} 
	a\preceq b\rightarrow c\mbox{ iff }a\wedge b\preceq c,
	\end{equation}
	for every $a,b,c\in A$. For each $a,b,c,d,e\in A$, if $a\rightarrow ((b\rightarrow c)\wedge (d\rightarrow e))=\top$ then  
	$a\rightarrow ((b\wedge d)\rightarrow (c\wedge e))=\top$.
\end{lemma}

\begin{proof}
	If $a\rightarrow ((b\rightarrow c)\wedge (d\rightarrow e))=\top$ then, by Lemma \ref{lem-BCI-ord} and \ref{PropC5},   
	$a\preceq b\rightarrow c$ and $a\preceq  d\rightarrow e$. By (\ref{eq-(*)}),  $a\wedge b\preceq c$ and $a\wedge d\preceq e$ and therefore, 
	$b\preceq a\rightarrow c$ and  $d\preceq a\rightarrow e$. So,   $b\wedge d\preceq a\rightarrow c$ and 
	$b\wedge d\preceq a\rightarrow e$. Thus, applying again (\ref{eq-(*)}), 
	$(b\wedge d)\wedge a\preceq c$ and $(b\wedge d)\wedge a\preceq e$. Hence, $(b\wedge d)\wedge a\preceq c\wedge e$. So, by (\ref{eq-(*)}) and 
	\ref{PropC5}, $a\rightarrow ((b\wedge d)\rightarrow (c\wedge e))=\top$.
	
\end{proof} 	

\begin{proposition}\label{PropInervalBCI}
	Let $\prt{A,\rightarrow,\top}$ be a BCI-Algebra such that $\prt{A,\preceq}$ is a meet-semilattice satisfying $(\ref{eq-(*)})$. Then  
	$\prt{\mathbb{A}, \Mapsto,[\top,\top]}$, where
	\begin{equation}\label{eq-mapsto}
	X\Mapsto Y=[(\underline{X}\rightarrow \underline{Y}) \wedge(\overline{X}\rightarrow \overline{Y}), \overline{X}\rightarrow \overline{Y}],
	\end{equation} is also a BCI-algebra which satisfies $(\ref{eq-(*)})$.
\end{proposition}

\begin{proof}
	Notice that in this case, defining  $X\preceq Y$ iff $X\Mapsto Y=[\top,\top]$, then  $X\preceq Y$  iff $\underline{X}\rightarrow \underline{Y}=\overline{X}\rightarrow \overline{Y}=\top$ iff $\underline{X}\preceq \underline{Y}$ and $\overline{X}\preceq \overline{Y}$.
	
	Thus, clearly, the properties \ref{Curry-I} to \ref{PropC5} are trivially satisfied. In the following,  \ref{PropC1} is proved.

	Since, $\prt{A,\rightarrow,\top}$ is a BCI-Algebra, by \ref{PropC1},  each $X,Y,Z\in \mathbb{A}$, 
	$(\underline{Y}\rightarrow \underline{Z})\rightarrow ((\underline{Z}\rightarrow \underline{X})\rightarrow 
	(\underline{Y}\rightarrow \underline{X}))=\top$ and  $(\overline{Y}\rightarrow \overline{Z})\rightarrow 
	((\overline{Z}\rightarrow \overline{X})\rightarrow (\overline{Y}\rightarrow \overline{X}))=\top$. Then,   
	$((\underline{Y}\rightarrow \underline{Z})\wedge (\overline{Y}\rightarrow \overline{Z}))\rightarrow ((\underline{Z}\rightarrow \underline{X})\rightarrow (\underline{Y}\rightarrow \underline{X}))=\top$ and  $((\underline{Y}\rightarrow \underline{Z})\wedge(\overline{Y}\rightarrow \overline{Z}))\rightarrow ((\overline{Z}\rightarrow \overline{X})\rightarrow (\overline{Y}\rightarrow \overline{X}))=\top$. So, by Lemma \ref{lem-BCI-ord}, $((\underline{Y}\rightarrow \underline{Z})\wedge (\overline{Y}\rightarrow \overline{Z}))\rightarrow ( ((\underline{Z}\rightarrow \underline{X})\rightarrow (\underline{Y}\rightarrow \underline{X}))\wedge ((\overline{Z}\rightarrow \overline{X})\rightarrow (\overline{Y}\rightarrow \overline{X})))=\top$.  Thus, by Lemma \ref{lem-BCI-ord2} and Eq. (\ref{eq-mapsto}), $\underline{Y\Mapsto Z}\rightarrow ( ((\underline{Z}\rightarrow \underline{X})\wedge (\overline{Z}\rightarrow \overline{X}))\rightarrow ((\underline{Y}\rightarrow \underline{X})\wedge (\overline{Y}\rightarrow \overline{X})))=\top$. Therefore by Eq. (\ref{eq-mapsto}), (*) $\underline{Y\Mapsto Z}\rightarrow (\underline{Z\Mapsto X}\rightarrow \underline{Y\Mapsto X})=\top$. On the other hand, by (C-1), $(\overline{Y}\rightarrow \overline{Z}) \rightarrow ((\overline{Z}\rightarrow \overline{X})\rightarrow(\overline{Y}\rightarrow\overline{X}))=\top$ and so, by Eq. (\ref{eq-mapsto}): (**)
	$\overline{Y\Mapsto Z}\rightarrow (\overline{Z\Mapsto X}\rightarrow \overline{Y\Mapsto X})=\top$. Thus, from 
	(*) and (**) and Lemma \ref{lem-BCI-ord},  $(\underline{Y\Mapsto Z}\wedge\overline{Y\Mapsto Z})\rightarrow (  (\underline{Z\Mapsto X}\rightarrow \underline{Y\Mapsto X})\wedge(\overline{Z\Mapsto X}\rightarrow \overline{Y\Mapsto X}))=\top$. Since, by Eq. (\ref{eq-mapsto}), $(\underline{Y\Mapsto Z}\wedge\overline{Y\Mapsto Z})=\underline{Y\Mapsto Z}$, then (***) $\underline{Y\Mapsto Z}\rightarrow (  (\underline{Z\Mapsto X}\rightarrow \underline{Y\Mapsto X})\wedge(\overline{Z\Mapsto X}\rightarrow \overline{Y\Mapsto X}))=\top$.  Thus, from (***) and (**),  $(\underline{Y\Mapsto Z}\rightarrow   \underline{(Z\Mapsto X)\Mapsto (Y\Mapsto X)})\wedge(\overline{Y\Mapsto Z}\rightarrow \overline{(Z\Mapsto X)\Mapsto (Y\Mapsto X)})=\top$ and $\overline{Y\Mapsto Z}\rightarrow \overline{(Z\Mapsto X)\Mapsto (Y\Mapsto X)}=\top$. Therefore, 
	by Eq. (\ref{eq-mapsto}), $(Y\Mapsto Z)\Mapsto((Z\Mapsto X)\Mapsto (Y\Mapsto X))=[\top,\top]$.

	\ref{PropC2}: Clearly, $\underline{X}\rightarrow \underline{Y} \geq ((\underline{X}\rightarrow \underline{Y}) 
	\wedge (\overline{X}\rightarrow \overline{Y}))$ and therefore, by (A-2), $((\underline{X}\rightarrow \underline{Y}) \wedge 
	(\overline{X}\rightarrow \overline{Y}))  \rightarrow \underline{Y} \geq (\underline{X}\rightarrow \underline{Y}) \rightarrow 
	\underline{Y}$. So, by \ref{BCI-Pr7} and \ref{PropC2},  $\underline{X}\rightarrow (((\underline{X}\rightarrow \underline{Y}) \wedge 
	(\overline{X}\rightarrow \overline{Y}))  \rightarrow \underline{Y})\geq \underline{X}\rightarrow ((\underline{X}\rightarrow 
	\underline{Y}) \rightarrow \underline{Y}) =\top$. So, by Eq. (\ref{eq-mapsto}) and \ref{BCI-Pr1},  $\underline{X}\rightarrow 
	(\underline{(X\Mapsto Y)}\rightarrow \underline{Y}) =\top$. On the other hand, by Eq. (\ref{eq-mapsto}) and \ref{PropC2},  (\#) 
	$\overline{X}\rightarrow (\overline{(X\Mapsto Y)}\rightarrow \overline{Y}) =\top$. Therefore, (\#\#) $(\underline{X}\rightarrow 
	(\underline{(X\Mapsto Y)}\rightarrow \underline{Y}))\wedge (\overline{X}\rightarrow (\overline{(X\Mapsto Y)}\rightarrow 
	\overline{Y})) =\top$. Hence, from (\#\#), (\#) and Eq. (\ref{eq-mapsto}): $\underline{X}\rightarrow 
	\underline{(X\Mapsto Y)\Mapsto Y}=\top$ and $\overline{X}\rightarrow \overline{(X\Mapsto Y)\Mapsto Y}=\top$. Consequently, $	
	[\underline{X}\rightarrow \underline{(X\Mapsto Y)\Mapsto Y} \wedge \overline{X}\rightarrow \overline{(X\Mapsto Y)\Mapsto Y},  
	\overline{X}\rightarrow \overline{(X\Mapsto Y)\Mapsto Y}]= [ \top, \top ]$. Therefore, by Eq. (\ref{eq-mapsto}), $X\Mapsto ((X\Mapsto Y)\Mapsto Y) = [ \top, \top]$.
	
	\underline{$\prt{\mathbb{A}, \Mapsto,[\top,\top]}$ is a meet-semilattice}. In fact, let $X,Y,Z\in \mathbb{A}$. Then, $X\Mapsto Y\wedge Z=[\top,\top]$ iff 
	$(\underline{X}\rightarrow(\underline{Y}\wedge \underline{Z}))\wedge (\overline{X}\rightarrow(\overline{Y}\wedge \overline{Z}))=\top$ and 
	$\overline{X}\rightarrow(\overline{Y}\wedge \overline{Z})=\top$ iff $\underline{X}\rightarrow\underline{Y}=\top$, 
	$\underline{X}\rightarrow\underline{Z}=\top$, $\overline{X}\rightarrow\overline{Y}=\top$ and $\overline{X}\rightarrow\overline{Z}=\top$ iff 
	$X\Mapsto Y=[\top,\top]$ and  $X\Mapsto Z=[\top,\top]$. \\
	In  addition, $X\preceq Y\Mapsto Z$ iff $\underline{X}\preceq (\underline{Y}\rightarrow \underline{Z})\wedge 
	(\overline{Y}\rightarrow \overline{Z})$ and $\overline{X}\preceq\overline{Y}\rightarrow \overline{Z}$ iff
	$\underline{X}\preceq \underline{Y}\rightarrow \underline{Z}$ and $\overline{X}\preceq \overline{Y}\rightarrow \overline{Z}$ iff
	$\underline{X}\wedge\underline{Y}\preceq\underline{Z}$ and $\overline{X}\wedge \overline{Y}\preceq \overline{Z}$ iff 
	$X\wedge Y \preceq  Z$. Therefore, $\prt{\mathbb{A}, \Mapsto,[\top,\top]}$ satisfies (\ref{eq-(*)}). 
	
\end{proof}

\vspace{1em}

If $A$ has two different elements, say $a$ and $b$, such that $a\rightarrow b=\top$, i.e. $a\preceq b$, then \textbf{$\Mapsto$ is not an interval representation of $\rightarrow$}. In particular, $[a,b]\Mapsto [a,b]=[\top,\top]$. Nevertheless, by \ref{PropC4}, $b\rightarrow a\neq \top$ and so $b\rightarrow a\not\in [a,b]\Mapsto [a,b]$. This leads us to the following \textit{general} theorem:

\begin{theorem}\label{TheoremIntervalBCInotRep}
	Let $\prt{A,\rightarrow,\top}$ be a BCI-Algebra. If there are $a,b\in A$ such that $a\neq b$ and $a\rightarrow b=\top$, then for any interval $\downmodels\in \mathbb{A}$ there is no interval representation  $\rightarrowtail$ for $\rightarrow$ such that  $\prt{\mathbb{A},\rightarrowtail,\downmodels}$ is a BCI-algebra.
\end{theorem}

\begin{proof}
	Case $\top\not\in\downmodels$. Then $a\rightarrow a=\top\not\in\downmodels = [a,a]\rightarrowtail [a,a]$ and therefore   $\rightarrowtail$  is not an interval representation of $\rightarrow$.
	
	Case  $\downmodels= [\top,\top]$, then $[a,b]\rightarrowtail [a,b]=\downmodels = [\top,\top]$. Nevertheless, by \ref{PropC4}, $b\rightarrow a\neq \top$ and so $b\rightarrow a\not\in [a,b]\rightarrowtail [a,b]$. Therefore, in this case  $\rightarrowtail$ also is not an interval representation of $\rightarrow$. 
	
	Case $\downmodels =[\alpha,\top]$ for some $\alpha<\top$. Then $\top\rightarrow \alpha\neq \top$. However, if  $\prt{\mathbb{A},\rightarrowtail,\downmodels}$ is a BCI-algebra then, by (A-8), $\downmodels\rightarrowtail [\top,\top]=[\top,\top]$ and therefore, $\top\rightarrow \alpha\not\in [\alpha,\top]\rightarrowtail [\top,\top]$ which means that $\rightarrowtail$ again is not an interval representation of $\rightarrow$. 
	
\end{proof}

\vspace{1em}

In the following, we propose a process for intervalization of BCI-algebras.

% \begin{proposition}
%\rot{Let $ \langle A, \rightarrow, \top \rangle $ be a BCI algebra, $ \langle A, \preceq \rangle $ a meet semilattice and $\mathbb{A}=\{[\underline{X},\overline{X}]: \underline{X},\overline{X}\in  A \mbox{ and } \underline{X}\preceq\overline{X}\}$. For  $X,Y\in\mathbb{A}$,  define $X \Twoheadrightarrow Y=[\overline{X}\rightarrow\underline{Y},\underline{X}\rightarrow\overline{Y}]$. Then $\Twoheadrightarrow $ is the best representation of $\rightarrow$. }
%\end{proposition}

% \begin{proof}
%	\rot{Citar um artigo.}
% \end{proof}

\begin{theorem}\label{ThIntBCI}
	Let $\prt{A,\rightarrow,\top}$ be a BCI-algebra, $\prt{A,\preceq}$ be a meet semilattice, such that for each 
	$x,y,z\in A$, $x\rightarrow (y\wedge z)= x \rightarrow y \wedge x \rightarrow z$ and 
	$\mathbb{A}=\{[\underline{X},\overline{X}]: \underline{X},\overline{X}\in  A \mbox{ and } \underline{X}\preceq\overline{X}\}$. 
	For $X,Y\in\mathbb{A}$,  define:
	\begin{quote} 
		\begin{enumerate}
			%	 	 \item $X\preceq Y$ if and only if $\underline{X}\leq\underline{Y}$ and $\overline{X}\leq\overline{Y}$
			\item $X\Twoheadrightarrow Y=[\overline{X}\rightarrow\underline{Y},\underline{X}\rightarrow\overline{Y}]$.
			\item $X\Rightarrow Y=[\underline{X}\rightarrow\underline{Y}\wedge \overline{X}\rightarrow\overline{Y},
			\underline{X}\rightarrow\overline{Y}].$
			%	 	 \item $X\ll Y$ if and only if $\overline{X}\leq\underline{Y}$.
		\end{enumerate}
	\end{quote}
	
	Then $\Twoheadrightarrow $ is the best representation of $\rightarrow$ and the structure $\prt{\mathbb{A},\Twoheadrightarrow,\Rightarrow,[\top,\top]}$ satisfies: 
	\begin{enumerate}[labelindent=\parindent, leftmargin=*,label=\normalfont{(IBCI\arabic*)}]
		\item $X\Twoheadrightarrow (Y\Twoheadrightarrow Z) = Y\Twoheadrightarrow (X\Twoheadrightarrow Z)$, \label{IBCI1}
		\item $X\Rightarrow (Y\Rightarrow Z) = Y\Rightarrow (X\Rightarrow Z)$,\label{IBCI2}
		\item $X\Twoheadrightarrow Y \precsim (Z\Twoheadrightarrow X)\Rightarrow (Z\Twoheadrightarrow Y)$,\label{IBCI3}
		\item $[\top, \top] \Twoheadrightarrow X = X$,\label{IBCI4}
		\item $X\ll Y \precsim Z\Longrightarrow X\ll Z$,\label{IBCI5}
		\item $X\precsim Y \ll Z\Longrightarrow X\ll Z$,\label{IBCI6}
		\item $X\precsim Y$ e $Y\precsim X\Longrightarrow X = Y$,\label{IBCI7}				
	\end{enumerate}
	
	where $X\ll Y \Longleftrightarrow X \Twoheadrightarrow Y = [\top, \top]$ e $X \precsim Y \Longleftrightarrow X \Rightarrow Y = [\top, \top]$.
	However, when $A$ has at least one  element different from $\top$, then $\prt{\mathbb{A},\Twoheadrightarrow,[\top,\top]}$ is not a BCI-algebra.
\end{theorem}

\begin{proof}   
	According to Proposition 4.4 at \cite{BS2013a} the operation $\Twoheadrightarrow $ is the best representation of $\rightarrow$.	Note that: 
	
	\begin{enumerate}
		\item $X\ll Y \Leftrightarrow X\Twoheadrightarrow Y=[\top,\top]\Leftrightarrow  \overline{X}\rightarrow \underline{Y}=\top$ and 
		$\underline{X}\rightarrow \overline{Y}=\top \Leftrightarrow \overline{X}\preceq\underline{Y}$ and $\underline{X}\preceq\overline{Y} 
		\Leftrightarrow \overline{X}\preceq\underline{Y}$.
		
		\item $X\precsim Y\Leftrightarrow X\Rightarrow Y=[\top,\top]\Leftrightarrow  
		\underline{X}\rightarrow \underline{Y}=\top$ and $\overline{X}\rightarrow \overline{Y}=\top$ and 
		$\underline{X}\rightarrow \overline{Y}=\top 
		\Leftrightarrow \underline{X}\preceq\underline{Y}$ and $\overline{X}\preceq\overline{Y}$ and 
		$\underline{X}\preceq\overline{Y} \Leftrightarrow \underline{X}\preceq\underline{Y}$   and 
		$\overline{X}\preceq\overline{Y}$, i.e. $\precsim$ 
		is the Kulisch-Miranker order.
	\end{enumerate}
	
	\ref{IBCI7} is satisfied, since $\prt{A,\preceq}$ is a poset and  ``$\precsim$'' is the Kulisch-Miranker order.
	
	Case of \ref{IBCI1}. $X\Rawint (Y\Rawint Z)=X\Rawint [\overline{Y}\rightarrow\underline{Z}, \underline{Y}\rightarrow\overline{Z}]= 
	[\overline{X}\rightarrow(\overline{Y}\rightarrow\underline{Z}), \underline{X}\rightarrow(\underline{Y}\rightarrow\overline{Z})]$. 
	According to property \ref{Curry-C} of BCI-algebras this term is also equal to 
	$[\overline{Y}\rightarrow(\overline{X}\rightarrow\underline{Z}), \underline{Y}\rightarrow(\underline{X}\rightarrow\overline{Z})]=
	Y\Rawint[\overline{X}\rightarrow\underline{Z}, \underline{X}\rightarrow\overline{Z}]=Y\Rawint (X\Rawint Z)$.
	
	Case of \ref{IBCI2}. $X \Rightarrow (Y \Rightarrow Z) = X \Rightarrow [\underline{Y} \rightarrow \underline{Z} \wedge \overline{Y} \rightarrow \overline{Z}, \underline{Y} \rightarrow \overline{Z} ] = [\underline{X}\rightarrow (\underline{Y} \rightarrow \underline{Z} \wedge \overline{Y} \rightarrow \overline{Z}) \wedge \overline{X} \rightarrow (\underline{Y} \rightarrow \overline{Z}), \underline{X} \rightarrow (\underline{Y} \rightarrow \overline{Z})]$. On the other hand, $Y\Rightarrow (X \Rightarrow Z) = Y \Rightarrow [\underline{X} \rightarrow \underline{Z} \wedge \overline{X} \rightarrow \overline{Z}, \underline{X} \rightarrow \overline{Z} ] = [\underline{Y}\rightarrow (\underline{X} \rightarrow \underline{Z} \wedge \overline{X} \rightarrow \overline{Z}) \wedge \overline{Y} \rightarrow (\underline{X} \rightarrow \overline{Z}), \underline{Y} \rightarrow (\underline{X} \rightarrow \overline{Z})] = [(\underline{Y}\rightarrow (\underline{X} \rightarrow \underline{Z}) \wedge \underline{Y}\rightarrow (\overline{X} \rightarrow \overline{Z})) \wedge \overline{Y} \rightarrow (\underline{X} \rightarrow \overline{Z}),  \underline{Y} \rightarrow (\underline{X} \rightarrow \overline{Z})]$. By property \ref{Curry-C},  the last  term is equal to: $[(\underline{X}\rightarrow (\underline{Y} \rightarrow \underline{Z}) \wedge \overline{X}\rightarrow (\underline{Y} \rightarrow \overline{Z})) \wedge \underline{X} \rightarrow (\overline{Y} \rightarrow \overline{Z}),  \underline{X} \rightarrow (\underline{Y} \rightarrow \overline{Z})]$ which is equal   \footnote{By associativity and commutativity of  meet.} to $[(\underline{X}\rightarrow (\underline{Y} \rightarrow \underline{Z}) \wedge \underline{X}\rightarrow (\overline{Y} \rightarrow \overline{Z})) \wedge \overline{X} \rightarrow (\underline{Y} \rightarrow \overline{Z}),  \underline{X} \rightarrow (\underline{Y} \rightarrow \overline{Z})] = [\underline{X}\rightarrow (\underline{Y} \rightarrow \underline{Z} \wedge \overline{Y} \rightarrow \overline{Z}) \wedge \overline{X} \rightarrow (\underline{Y} \rightarrow \overline{Z}), \underline{X} \rightarrow (\underline{Y} \rightarrow \overline{Z})] = X \Rightarrow ( Y \Rightarrow Z).$ 
	
	Case of \ref{IBCI3}. By definition, $X\Twoheadrightarrow Y=[\overline{X}\rightarrow\underline{Y},\underline{X}\rightarrow\overline{Y}]$ and 
	$(Z\Twoheadrightarrow X)\Rightarrow(Z\Twoheadrightarrow Y) = [(\overline{Z} \rightarrow \underline{X}) \rightarrow (\overline{Z} \rightarrow \underline{Y}) \wedge (\underline{Z} \rightarrow \overline{X}) \rightarrow (\underline{Z} \rightarrow \overline{Y}) , (\overline{Z} \rightarrow \underline{X}) \rightarrow (\underline{Z} \rightarrow \overline{Y})]$. Since, by \ref{BCI-Pr2}, \ref{BCI-Pr7} and \ref{BCI-Pr6}, 
	$\overline{X}\rightarrow\underline{Y}\preceq  \underline{X}\rightarrow\underline{Y}\preceq 
	(\overline{Z}\rightarrow\underline{X})\rightarrow (\overline{Z}\rightarrow\underline{Y})$ and 
	$\overline{X}\rightarrow\underline{Y}\preceq  \overline{X}\rightarrow\overline{Y}\preceq
	(\underline{Z}\rightarrow\overline{X})\rightarrow (\underline{Z}\rightarrow\overline{Y})$, then 
	$\overline{X}\rightarrow\underline{Y}\preceq  (\overline{Z}\rightarrow\underline{X})\rightarrow (\overline{Z}\rightarrow\underline{Y}) \wedge
	(\underline{Z}\rightarrow\overline{X})\rightarrow (\underline{Z}\rightarrow\overline{Y})$. On the other hand, by \ref{BCI-Pr2}, \ref{BCI-Pr7} and \ref{BCI-Pr6},  
	$\underline{X}\rightarrow\overline{Y}\preceq  (\overline{Z}\rightarrow\underline{X})\rightarrow (\overline{Z}\rightarrow\overline{Y})
	\preceq (\overline{Z} \rightarrow \underline{X}) \rightarrow (\underline{Z}\rightarrow\overline{Y})$. 
	Therefore,  $X\Twoheadrightarrow Y\precsim (Z\Twoheadrightarrow X)\Rightarrow(Z\Twoheadrightarrow Y)$.
	
	Case of \ref{IBCI4}. By (A-8), $[\top,\top]\Rawint [\underline{X},\overline{X}]=[\top\rightarrow\underline{X},\top\rightarrow\overline{X}]=
	[\underline{X},\overline{X}]$.
	
	% Before we proceed, observe that: $X\ll Y$ is equivalent to $\overline{X}\rightarrow\underline{Y}=\top$,  which is also equivalent to $X\Rawint Y=[\top,\top]$.
	
	%Case of \ref{PropCI3}. Suppose $X\ll Y$, then $\underline{X}\preceq\overline{X}\preceq\underline{Y}\preceq\overline{Y}$; i.e. $X\precsim Y$.
	
	Case of \ref{IBCI5} and \ref{IBCI6}. Suppose $X\ll Y\precsim Z$, then $\overline{X}\preceq\underline{Y}\preceq\underline{Z}$. Hence, $X\ll Z$. \ref{IBCI6} is analogous.
	
	%Therefore, the structure $\prt{\mathbb{A},\Twoheadrightarrow,\Rightarrow,[\top,\top]}$ \textbf{is an SBCI algebra}.
	
	For each BCI  $\prt{A,\rightarrow,\top}$ with at least an element $a\in A$ such that $a\neq \top$, the algebra  
	$\prt{\mathbb{A},\Rawint,[\top,\top]}$ \textbf{is not a BCI-algebra}. In fact, it fails to satisfy \ref{Curry-I}, 
	since $[a\wedge \top,\top]\Rawint [a\wedge \top,\top]=[\top\rightarrow a\wedge \top,a\wedge \top\rightarrow\top]=
	[a\wedge \top, a\wedge \top\rightarrow\top]\neq [\top,\top]$. 
	
\end{proof}

%In the above theorem, observe that by Proposition \ref{pro-BCI-BCK}, $\prt{A,\rightarrow,\top}$  is a BCK algebra. \textcolor{red}{(neste caso, a intervalização não foi feita para BCK em vez de BCI?)} 
Observe that $\Twoheadrightarrow$ is the best interval representation of $\rightarrow$, but $\Rightarrow$ is not an interval representation of $\rightarrow$.

%\begin{remark}
%	\textcolor{red}{Observe that $[\underline{X}, \overline{X}] \Twoheadrightarrow [\top, \top] = [\overline{X}\rightarrow \top,\underline{X}\rightarrow \top] = [\top, \top]$, since $\prt{A,\rightarrow,\top}$  is a BCK-algebra, therefore $[\underline{X}, \overline{X}] \ll [\top, \top]$, this is $\prt{\mathbb{A},\Twoheadrightarrow,\Rightarrow,[\top,\top]}$ is an SBCK-algebra}
%\end{remark}  

\begin{definition}
	Given a BCI(K)-algebra $\prt{A,\rightarrow,\top}$,  the structure $\prt{\mathbb{A},\Twoheadrightarrow,\Rightarrow,[\top,\top]}$ obtained by the method used in Theorem \ref{ThIntBCI} is called 
	\textbf{Interval BCI(K)-algebra}, \textbf{IBCI} (\textbf{IBCK}).
\end{definition}

The process of intervalization destroys some basic properties of BCI-algebras, like \textbf{(OP)}, and  some properties are generalized.
%such as \ref{Curry2-C}, \ref{CIScndPIsot}, \ref{CILeftNeut} and \textcolor{red}{\ref{CIMP}}. 

\begin{theorem}
	Given an IBCI(K)-algebra, $\mathbb{A}$, and $X,Y,Z\in\mathbb{A}$, the following properties are satisfied:
	
	\begin{enumerate}[labelindent=\parindent, leftmargin=*,label=($\mathcal{G}$-\arabic*)]
		%	 	\item $(Y\Rightarrow Z)\preceq
		%	 	((Z\Rightarrow X)\Rightarrow(Y\Rightarrow X))$; whenever $Z\ll X$.
		%	 	\item $X\Rightarrow ((X\Rightarrow Y)\Rightarrow Y)=[\top,\top]$; whenever $X\ll Y$.	 	
		%	 	\item $X_d\Rightarrow ((X_d\Rightarrow Y_d)\Rightarrow Y_d)=[\top,\top]$
		\item $X\Twoheadrightarrow X=[\overline{X}\rightarrow\underline{X},\top]$. \label{Curry3-I}
		\item $X\Twoheadrightarrow Y=[\top,\top]$ iff $\overline{X}\preceq\underline{Y}$.
	\end{enumerate}
\end{theorem}

\begin{proof}
	\begin{enumerate}[labelindent=\parindent, leftmargin=*,label=($\mathcal{G}$-\arabic*)]

		\item $X\Twoheadrightarrow X=[\overline{X}\rightarrow\underline{X},\underline{X}\rightarrow\overline{X}]
		\stackrel{\mbox{\ref{PropC5}}}{=}[\overline{X}\rightarrow\underline{X},\top]$.
		
		\item  $X\Twoheadrightarrow Y=[\top,\top]$  iff $\overline{X}\rightarrow\underline{Y}=\top$ and
		$\underline{X}\rightarrow\overline{Y}=\top$ iff $\overline{X}\preceq\underline{Y}$.
		
	\end{enumerate}
	
\end{proof}

\begin{corollary}\label{CorolXiXDeg}
	$X\Twoheadrightarrow X=[\top,\top]$ iff $X$ is degenerate.
\end{corollary}

We conclude that the relation ``$\ll$''  corresponding to the operator ``$\Twoheadrightarrow$'' will be reflexive (and hence a partial order) only if it is restricted to the subset of the degenerate intervals of $\mathbb{A}$.

The next proposition provides another situation in which an intervalized BCI-algebra behaves like a BCI.

\begin{proposition}
	Given an IBCI(K)-algebra, $\mathbb{A}$, and $X,Y,Z,X_d=[u,u], Y_d=[v,v]\in\mathbb{A}$, the following properties are satisfied:
	
	\begin{enumerate}[labelindent=\parindent, leftmargin=*,label=($\mathcal{C}_d$-\arabic*)]
		\item $(Y\Twoheadrightarrow Z)\precsim
		((Z\Twoheadrightarrow X_d)\Twoheadrightarrow(Y\Twoheadrightarrow X_d))$, i.e.,\\
		$(Y\Twoheadrightarrow Z)\Rightarrow ((Z\Twoheadrightarrow X_d)\Twoheadrightarrow(Y\Twoheadrightarrow X_d))=[\top,\top]$;
		\item $X_d\Twoheadrightarrow ((X_d\Twoheadrightarrow Y_d)\Twoheadrightarrow Y_d)=[\top,\top]$; 
		% 	 	\item $X_d\Twoheadrightarrow X_d=[\top,\top]$; \label{Curry4-I} --- Corollary \ref{CorolXiXDeg}.
		\item $X\Twoheadrightarrow X_d=X\Rightarrow X_d$. \label{Curry4b-I}
	\end{enumerate}
\end{proposition}

\begin{proof}
	\begin{enumerate}[labelindent=\parindent, leftmargin=*,label=$\mathcal{C}_d$-\arabic*]
		
		\item $(Y\Twoheadrightarrow Z)=[\overline{Y}\rightarrow \underline{Z}, \underline{Y} \rightarrow \overline{Z}]$. Moreover, $((Z\Twoheadrightarrow X_d)\Twoheadrightarrow(Y\Twoheadrightarrow X_d))=[(\underline{Z}\rightarrow \overline{X_d})\rightarrow(\overline{Y}\rightarrow \underline{X_d}), (\overline{Z}\rightarrow \underline{X_d})\rightarrow(\underline{Y}\rightarrow \overline{X_d})]=[(\underline{Z}\rightarrow u)\rightarrow(\overline{Y}\rightarrow u),(\overline{Z}\rightarrow u)\rightarrow(\underline{Y}\rightarrow u)]$. By property \ref{PropC1}, $\overline{Y}\rightarrow \underline{Z}\preceq (\underline{Z}\rightarrow u)\rightarrow(\overline{Y}\rightarrow u)$ and 
		$\underline{Y}\rightarrow \overline{Z}\preceq (\overline{Z}\rightarrow u)\rightarrow(\underline{Y}\rightarrow 
		u)$. Therefore, $(Y\Rightarrow Z)\precsim ((Z\Rightarrow X_d)\Rightarrow(Y\Rightarrow X_d))$.
		
		\item $X_d\Twoheadrightarrow ((X_d\Rightarrow Y_d)\Twoheadrightarrow Y_d)= X_d\Twoheadrightarrow 
		([\overline{X_d}\rightarrow v,\underline{X_d}\rightarrow v ]\Rightarrow Y_d)=X\Rightarrow ([(\underline{X_d}\rightarrow v)
		\rightarrow v,(\overline{X_d}\rightarrow v)\rightarrow v])=[\overline{X_d}\rightarrow
		((\underline{X_d}\rightarrow v)\rightarrow v),\underline{X_d}\rightarrow((\overline{X_d}\rightarrow v)\rightarrow v)]=
		[u\rightarrow((u\rightarrow v)\rightarrow v),u\rightarrow ((u\rightarrow v)\rightarrow v)]=[\top,\top]$.
		
		% 		 	\item See \azul{Corollary} \ref{CorolXiXDeg}.
		
		\item $X\Twoheadrightarrow X_d=[\overline{X}\rightarrow u,\underline{X}\rightarrow u]=
		[\underline{X}\rightarrow u \wedge \overline{X}\rightarrow u, \underline{X}\rightarrow u] =
		X\Rightarrow X_d$.
	\end{enumerate}
	
\end{proof}

In what follows  a list of properties of  IBCI-algebras is provided.

\begin{theorem}
	An IBCI-algebra has the following properties: For all $X_d=[u,u], Z_d=[w,w],X,Y,Z\in\mathbb{A}$,
	
	\begin{enumerate}[labelindent=\parindent, leftmargin=*,label=\normalfont{(B-\arabic*)}]
		
		\item $[\top,\top]\precsim X$ implies $X=[\top,\top]$,
		
		\item $X \precsim Y$ implies $Y\Twoheadrightarrow Z\precsim X\Twoheadrightarrow Z$, 
		
		\item $X\precsim Y$ implies $Z\Twoheadrightarrow X\precsim Z\Twoheadrightarrow Y$,
		
		\item $X\precsim Y$ and $Y\precsim Z$ implies $X\precsim Z$,
		
		\item $X\precsim Y\Twoheadrightarrow Z_d$ implies $Y\precsim X\Twoheadrightarrow Z_d$,
		
		\item $X\Twoheadrightarrow Y\precsim (Z_d\Twoheadrightarrow X)\Twoheadrightarrow (Z_d\Twoheadrightarrow Y)$,\label{Curry3-B}
		
		\item $((Y\Twoheadrightarrow X_d)\Twoheadrightarrow X_d)\Twoheadrightarrow X_d=Y\Twoheadrightarrow X_d$,
		
		\item $X\Twoheadrightarrow Y\precsim (Y\Twoheadrightarrow X)\Twoheadrightarrow [\top,\top]$,
		
		\item $(X\Twoheadrightarrow Y)\Twoheadrightarrow [\top,\top]=(X\Twoheadrightarrow [\top,\top])\Twoheadrightarrow (Y\Twoheadrightarrow
		[\top,\top])$. 
	\end{enumerate}
\end{theorem}

\begin{proof}
	\begin{enumerate}[labelindent=\parindent, leftmargin=*,label=\normalfont{(B-\arabic*)}]
		\item Suppose $[\top,\top]\precsim X$, then $\top\precsim\underline{X}$ and $\top\precsim\overline{X}$. By \ref{BCI-Pr1}, 
		$\underline{X}=\top$ and $\overline{X}=\top$. Therefore, $X=[\top,\top]$.
		
		\item Follows from the  relation  $\precsim$ and \ref{BCI-Pr2}.
		
		\item Follows from the  relation  $\precsim$ and \ref{BCI-Pr7}.
		
		\item Straightforward.
		
		\item Suppose $X\precsim Y\Twoheadrightarrow Z_d$, then $\underline{X}\preceq \overline{Y}\rightarrow w$ and $\overline{X}\preceq \underline{Y}\rightarrow w$. 
		By \ref{BCI-Pr5} $\overline{Y}\preceq \underline{X}\rightarrow w$ and $\underline{Y}\preceq \overline{X}\rightarrow w$. Hence, $Y\precsim X\Twoheadrightarrow Z_d$.
		
		\item $(X\Twoheadrightarrow Y)\precsim (Z_d\Twoheadrightarrow X)\Twoheadrightarrow (Z_d\Twoheadrightarrow Y)$ iff 
		$(X\Twoheadrightarrow Y)\precsim [w\rightarrow \underline{X},w\rightarrow \overline{X}]\Twoheadrightarrow 
		[w\rightarrow \underline{Y},w\rightarrow \overline{Y}]$ iff
		$(X\Twoheadrightarrow Y)\precsim [(w\rightarrow \overline{X})\raw (w\raw\underline{Y}),(w\rightarrow 
		\underline{X})\raw(w\rightarrow \overline{Y})]$ iff $[\overline{X}\rightarrow\underline{Y},\underline{X}\rightarrow\overline{Y}]
		\precsim [(w\rightarrow\overline{X})\rightarrow (w\rightarrow\underline{Y}),(w\rightarrow\underline{X})\raw (w\raw\overline{Y})]$.
		By Property \ref{BCI-Pr6}, $\overline{X}\raw\underline{Y}\preceq (w\rightarrow\overline{X})\rightarrow (w\rightarrow\underline{Y})$ 
		and $\underline{X}\raw\overline{Y}\preceq (w\rightarrow\underline{X})\rightarrow (w\rightarrow\overline{Y})$.
		
		\item $((Y\Twoheadrightarrow X_d)\Twoheadrightarrow X_d)\Twoheadrightarrow X_d=([\overline{Y}\rightarrow\underline{X_d},
		\underline{Y}\rightarrow\overline{X_d}]\Twoheadrightarrow X_d)\Twoheadrightarrow X_d=
		[(\underline{Y}\rightarrow u)\rightarrow u,(\overline{Y}\raw u)\raw u]\Twoheadrightarrow X_d=
		[((\overline{Y}\raw u)\raw u)\raw u,((\underline{Y}\raw u)\raw u)\raw u]\stackrel{\ref{BCI-Pr9}}{=}
		[\overline{Y}\raw u,\underline{Y}\rightarrow u]=[\overline{Y}\raw\underline{X_d},\underline{Y}\rightarrow\overline{X_d}]=
		Y\Twoheadrightarrow X_d$.
		
		\item Since $\overline{X}\rightarrow\underline{Y}\preceq (\underline{Y}\rightarrow\overline{X})\rightarrow\top$ and 
		$\underline{X}\rightarrow\overline{Y}\preceq (\overline{Y}\rightarrow\underline{X})\rightarrow\top$, then by definition 
		$X\Twoheadrightarrow Y\precsim (Y\Twoheadrightarrow X)\Twoheadrightarrow [\top, \top]$.
		
		\item $(X\Twoheadrightarrow [\top,\top])\Twoheadrightarrow (Y\Twoheadrightarrow [\top,\top])=
		[\overline{X}\raw\top,\underline{X}\raw\top]\Twoheadrightarrow [\overline{Y}\raw\top,\underline{Y}\rightarrow\top]=
		[(\underline{X}\rightarrow\top)\rightarrow (\overline{Y}\rightarrow\top),(\overline{X}\rightarrow \top)\rightarrow
		(\underline{Y}\rightarrow\top)]\stackrel{\ref{BCI-Pr11}}{=}[(\underline{X}\raw\overline{Y})\raw\top,(\overline{X}\raw
		\underline{Y})\raw\top]=[\overline{X}\raw\underline{Y},\underline{X}\raw\overline{Y}]\Twoheadrightarrow [\top,\top]=
		(X\Twoheadrightarrow Y)\Twoheadrightarrow [\top,\top]$.
		
	\end{enumerate}	
	
\end{proof}

\begin{proposition}[\textbf{r-WOP}]\label{PropLeftOP1}
	If $X\precsim Y$, then $X\Twoheadrightarrow Y=[\overline{X}\rightarrow\underline{Y},\top]$. 
\end{proposition}

\begin{proof}
	Suppose $X\precsim Y$, then $\underline{X}\preceq\underline{Y}\preceq\overline{Y}$. Therefore $\underline{X}\rightarrow\overline{Y}=\top$ and hence
	$X\Twoheadrightarrow Y=[\overline{X}\rightarrow\underline{Y},\top]$.
	
\end{proof}

\vspace{1em}

\begin{theorem}
	The properties:
	
	\begin{enumerate}[labelindent=\parindent, leftmargin=*,label=($\mathbf{OP}_{M_\arabic*}$)]
		\item $X\ll Y\mbox{ if and only if } X\Rightarrow Y=[\top,\top]$ and	
		\item $X\precsim Y\mbox{ if and only if } X\Twoheadrightarrow Y=[\top,\top]$
	\end{enumerate}
	
	are not satisfied. However, the following holds:
	
	\begin{enumerate}[labelindent=\parindent, leftmargin=*,label=($\mathbf{OP}_{\alph*}$)]
		\item $\mbox{\ If $X\ll Y$, then $X\Rightarrow Y=[\top,\top]$}$ and	
		\item $\mbox{\ If $X\Twoheadrightarrow Y=[\top,\top]$, then $X\precsim Y$.}$
	\end{enumerate}
\end{theorem}

\begin{proof}
	$\mathbf{(OP_{M_1})}$ is not satisfied. If fact, if $ X\Rightarrow Y=[\top,\top] $ then $\underline{X}\preceq\underline{Y}$ and $\overline{X} \preceq\overline{Y}$, what does not mean that $\overline{X} \preceq\underline{Y}$.  $\mathbf{(OP_{M_2})}$ is not satisfied, since for  a given non-degenerate interval $X=[\underline{X},\overline{X}]$,  
	$X\precsim X$ and $X \Twoheadrightarrow X = [\overline{X} \raw \underline{X}, \underline{X}\raw\overline{X}]$ is not necessarily equal to $[\top,\top]$.
	
	\textbf{(OPa)} Suppose $\overline{X}\preceq\underline{Y}$, then $\underline{X}\preceq\overline{X}\preceq\underline{Y}\preceq\overline{Y}$. Therefore, $X\Rightarrow Y=[\top,\top]$.
	
	\textbf{(OPb)} Suppose $X\Twoheadrightarrow Y=[\top,\top]$, then $[\overline{X}\raw\underline{Y},\underline{X}\rightarrow\overline{Y}]=[\top,\top]$; i.e. 
	$\overline{X}\rightarrow\underline{Y}=\top$ and $\underline{X}\rightarrow\overline{Y}=\top$, so, $\overline{X}\preceq\underline{Y}$. Therefore $\underline{X}\preceq\underline{Y}$ and $\overline{X}\preceq\overline{Y}$.
	
\end{proof}

\begin{proposition}
	The implications map degenerate intervals to degenerate intervals. 
\end{proposition}

\begin{proof}
	Straightforward, since $[u,u]\Rightarrow [v,v]=[u\rightarrow v,u\rightarrow v]=[w,w]$ and by \ref{Curry4b-I}
	$[u,u]\Twoheadrightarrow[v,v]=[w,w]$.
	
\end{proof}

\vspace{1em}

%The next proposition shows how we can build an implication which satisfies \textbf{(OP)} from an arbitrary SBCI algebra .

%\begin{proposition}
%	Let $\prt{A,\twoheadrightarrow,\rightarrow,\top}$ be a SBCI-algebra, the SBCI-algebra $\prt{A,\twoheadrightarrow,\rightarrow_*,\top}$, such that:

%		  \begin{equation*}
%		    x \rightarrow_* y = 
%%		          \top \text{, if $x\preceq y$}
%		          \\
%		          x\rightarrow y \text{, otherwise}
%		    \end{cases}
%		   \end{equation*}
%	satisfies \textbf{(OP)}.
%\end{proposition}

%\begin{proof}
%	If $x\rightarrow_* y=\top$ then $x\preceq y$ or \rot{ $x\rightarrow y=\top$.Therefore, $x\preceq y$.}
%\end{proof}

%Note that if $\prt{A,\rightarrow,\top}$ is a BCI-algebra, then the implications ``$\raw$'' and ``$\raw_*$'' coincide. the proof is simple: if 
%$x\not\preceq y$, then $x\raw_* y=x\raw y$, else $x\rightarrow y=\top=x\raw_* y$.

%\begin{proposition}
%	If $\prt{\mathbb{A},\Twoheadrightarrow, \Rightarrow,[\top,\top]}$ is an IBCI obtained from a BCI algebra $\prt{A,\rightarrow,\top}$, then the implication ``$\Rightarrow_*$'' is not always correct.
%\end{proposition}

%\begin{proof}
%	Take $X,Y\in\mathbb{A}$, such that $X\precsim Y$, by
%\rot{proposition \ref{PropLeftOP1} 
%	$X\Rightarrow Y=[\overline{X}\rightarrow\underline{Y},\top]$ (a proposição 5.3 se refere a implicação dupla, por isso essa igualdade não é verdadeira)} and $(X\Rightarrow_{*} Y)=[\top,\top]$. $X\Rightarrow_* Y$ is correct
%only if $\overline{X}\rightarrow\underline{Y}=\top$
%\end{proof}

Therefore, the mathematical structure that arises from the intervalization of a BCI-algebra is a new mathematical structure which deserves to be developed. This structure will be called here \textbf{Semi-BCI algebra} and is what is exposes from now on.

\section{Semi-BCI Algebra }\label{Sec-SBCIA}

This paper showed that some implications do not satisfy the order property \textbf{(OP)} and  the correct intervalization of structures leads to relaxed structures. This section proposes a new algebraic structure which aims to capture both situations.

\begin{definition}[Semi-BCI (SBCI) algebra]
	Given a set $A$ endowed with two binary operations: ``$\twoheadrightarrow$'' and ``$\rightarrow$'', a structure 
	$\langle A,\twoheadrightarrow, \rightarrow,\top\rangle$ is a \textbf{Semi-BCI (SBCI) Algebra} whenever for all $x,y,z\in A$,
	
	\begin{enumerate}[labelindent=\parindent, leftmargin=*,label=\normalfont{(SBCI\arabic*)}]
		\item $x\twoheadrightarrow (y\twoheadrightarrow z)=y\twoheadrightarrow (x\twoheadrightarrow z)$\label{Curry2-C},
		\item $x\rightarrow (y\rightarrow z)=y\rightarrow (x\rightarrow z)$\label{Curry2b-C},
		%	 \item $x\preceq y\Longrightarrow z\twoheadrightarrow x\preceq z\twoheadrightarrow y$\label{CIScndPIsot}
		\item $x\twoheadrightarrow y\preceq (z\twoheadrightarrow x)\rightarrow (z\twoheadrightarrow y)$\label{CIScndPIsot},
		%\item \rot{$x\rightarrow y\ll (z\rightarrow x)\twoheadrightarrow (z\rightarrow y)$\label{CIScndPIsot2}}
		\item $\top\twoheadrightarrow x=x$\label{CILeftNeut},
		\item If $x\ll y \preceq z$ then $x\ll z$	\label{PropCI4},
		\item If $x\preceq y \ll z$ then $x\ll z$	\label{PropCI5},
		\item If $x\preceq y$ and $y\preceq x$ then $x=y$. \label{PropCI6}
	\end{enumerate}
	
	Where $x\ll y\Leftrightarrow x\twoheadrightarrow y=\top$ and $x\preceq y\Leftrightarrow x\rightarrow y=\top$.  
	
	\vspace{1em}
	A SBCI-algebra which satisfies: $x\ll \top$ for all $x\in A$,  is called \textbf{Semi-BCK (SBCK) algebra }. An element $x$ of a SBCI-algebra which satisfies: $x\ll x$, is called \textbf{total}.
\end{definition}

The  the operation $\twoheadrightarrow$ (operation $\rightarrow$) is said to satisfy the  \textbf{Weak Order Property (WOP)} whenever the relation ``$\ll$'' (relation ``$\preceq$'')  is not a partial order.

\begin{example}\label{ExSBCI}
	Consider the following structure: $\mathbf{A=\prt{[0,1],\twoheadrightarrow_R,\rightarrow_{LK}, 1}}$, such that $x\twoheadrightarrow_R y=1-x+xy$, 
	see Reichenbach \cite{Reichenbach1935}, and $x\rightarrow_{LK} y= \min\{1, 1 - x + y \}$. Thus, since $\rightarrow_{LK}$ is the \L ukasiewicz implication, 
	then the corresponding relation $\leq$ is the usual order and $x\ll y$ if and only if $x\twoheadrightarrow_R y=1$ if and only if $x=0$ or $y=1$.
	It is straightforward to check the satisfaction of \ref{PropCI4}-\ref{PropCI6}. \textbf{\ref{Curry2-C}} 
	$x\twoheadrightarrow_R(y\twoheadrightarrow_R z)=1-xy+xyz=y\twoheadrightarrow_R(x\twoheadrightarrow_R z)$. 
	\ref{Curry2b-C} For $x \rightarrow_{LK} (y \rightarrow_{LK} z ) = min(1,1 - x + min(1, 1 - y + z))$ and $y \rightarrow_{LK} (x \rightarrow_{LK} z ) = min(1,1 - y + min(1, 1 - x + z)),$ let be the following cases: \textbf{(1)} If $x \rightarrow_{LK} (y \rightarrow_{LK} z ) = 1$, then $1 - x + min(1, 1 - y + z) \geq 1 \Leftrightarrow min(1, 1 - y + z) \geq x \Leftrightarrow 1 - y + z \geq x \Leftrightarrow 1 - x + z \geq y \Leftrightarrow min(1, 1 - x + z) \geq y \Leftrightarrow 1 - y + min(1, 1 - x + z) \geq 1 \Leftrightarrow y \rightarrow_{LK} (x \rightarrow_{LK} z ) =1 .$ \textbf{(2)} If $x \rightarrow_{LK} (y \rightarrow_{LK} z ) \neq 1$, then  $1 - x + min(1, 1 - y + z) < 1 \Leftrightarrow min(1, 1 - y + z) < x \Leftrightarrow 1 - y + z < x \Leftrightarrow 1 - x + z < y \Leftrightarrow min(1, 1 - x + z) < y \Leftrightarrow 1 - y + min(1, 1 - x + z) < 1.$ Therefore, $ y \rightarrow_{LK} (x \rightarrow_{LK} z ) = 1 - y + min(1, 1 - x + z) = 1 - y + 1 - x + z = 1 - x + 1 - y + z = 1 - x + min(1, 1 - y + z) = x \rightarrow_{LK} (y \rightarrow_{LK} z ) .$ Concerning  \ref{CIScndPIsot}: $x\twoheadrightarrow_R y=1-x+xy$ and
	$(z\twoheadrightarrow_R x)\rightarrow_{LK} (z\twoheadrightarrow_R y)=\min\{1, 1 - (1-z+zx) + (1-z+zy) \}=
	\min\{1,1-zx + zy\}$.  Since, $1-x+xy\leq 1-zx + zy$ then $1-x+xy\leq_{LK} \min\{1,1-zx + zy\}$ 
	%% observar que x-xy\geq x-y \geq z(x-y) =zx-zy e portanto 1-x+xy\leq 1-zx+zy
	and therefore $(x\twoheadrightarrow_R y)\leq_{LK} ( (z\twoheadrightarrow_R x)\rightarrow_{LK} (z\twoheadrightarrow_R y))$
	\ref{CILeftNeut} $1\twoheadrightarrow_R y=1-1+1\cdot y=y$. 
	Therefore $A=\prt{[0,1],\twoheadrightarrow_R,\rightarrow_{LK}, 1}$ is an SBCI-algebra.
\end{example}

\begin{remark}
	Note that any IBCI-algebra is a SBCI-algebra.
\end{remark}

\begin{proposition}\label{prop-CI7-CI8}
	In a SBCI-algebra, $\langle A,\twoheadrightarrow, \rightarrow,\top\rangle$, the following hold:
	\begin{enumerate}[labelindent=\parindent, leftmargin=*,label=\normalfont{(SBCI\arabic*)}]
		\setcounter{enumi}{7}
		\item If $x\ll y$ and $y\ll z$ then $ x\ll z$, \label{PropCI09}
		\item If $x\ll y$ and $y\ll x$ then $ x=y$,\label{PropCI10}
		%  		\item $y\ll z\Longrightarrow x\twoheadrightarrow y\preceq x\twoheadrightarrow z$. \label{PropCI11n}
		\item $(y\twoheadrightarrow z)\preceq ((z\twoheadrightarrow  x)\rightarrow (y\twoheadrightarrow x))$, \label{PropCI12n}
		\item If $\top \preceq x$, then $ x  = \top $, \label{PropCI14}
		\item $x\rightarrow x=\top$\label{PropCI1}
		\item If $x\ll y$ then $x\preceq y$ \label{PropCI3},
		\item $x \twoheadrightarrow y \preceq x \rightarrow y$, \label{PropCI15}
		\item $ x \rightarrow ((x\rightarrow y)\rightarrow y) = \top $, \label{PropCI13}
		\item If $x\ll y$ then $x\twoheadrightarrow ((x\twoheadrightarrow y)\twoheadrightarrow y)=\top$, \label{CIMP}
		\item If $x\ll y$, then $z\twoheadrightarrow x \preceq z \twoheadrightarrow y$,\label{SBCI17}
		\item If $x\ll y$, then $y \twoheadrightarrow z \preceq x \twoheadrightarrow z$. \label{SBCI18}
	\end{enumerate}
\end{proposition}

\begin{proof}
	The proof of items \ref{PropCI09}-\ref{PropCI12n} and \ref{PropCI15}-\ref{CIMP} are straightforward. \ref{PropCI14} By \ref{CIScndPIsot}, $(\top \twoheadrightarrow x) \rightarrow ((\top \twoheadrightarrow \top) \rightarrow (\top \twoheadrightarrow x)) =\top \stackrel{\ref{CILeftNeut}}{\Rightarrow} x \rightarrow (\top \rightarrow x) = \top$, as $\top \rightarrow x = \top$, then $x \rightarrow \top = \top$. So, by \ref{PropCI6}, $x = \top$. \ref{PropCI1} By \ref{PropCI12n}, $(\top \twoheadrightarrow \top) \rightarrow ((\top \twoheadrightarrow x) \rightarrow (\top \twoheadrightarrow x)) =\top \stackrel{\ref{CILeftNeut}}{\Rightarrow} \top \rightarrow (x \rightarrow x) = \top \stackrel{\ref{PropCI14}}{\Rightarrow} x \rightarrow x = \top$. \ref{PropCI3} By \ref{CIScndPIsot}, $(x \twoheadrightarrow y) \rightarrow ((\top \twoheadrightarrow x) \rightarrow (\top \twoheadrightarrow y)) = \top$, but as $x \twoheadrightarrow y = \top$ then $\top \rightarrow ((\top \twoheadrightarrow x) \rightarrow (\top \twoheadrightarrow y)) = \top \stackrel{\ref{CILeftNeut}}{\Rightarrow} \top \rightarrow (x \rightarrow y) = \top \stackrel{\ref{PropCI14}}{\Rightarrow} x \rightarrow y = \top$, so $x \preceq y$. \ref{SBCI17} By \ref{CIScndPIsot}, $(x \twoheadrightarrow y) \rightarrow ((z \twoheadrightarrow x) \rightarrow (z \twoheadrightarrow y)) =\top$, as $x \twoheadrightarrow y = \top$ then by \ref{PropCI14}, $(z \twoheadrightarrow x) \rightarrow (z \twoheadrightarrow y) = \top$, thus $(z \twoheadrightarrow x) \preceq (z \twoheadrightarrow y)$; \ref{SBCI18} By \ref{PropCI12n}, $(x \twoheadrightarrow y) \rightarrow ((y \twoheadrightarrow z) \rightarrow (x \twoheadrightarrow z)) =\top$, as $x \twoheadrightarrow y = \top$ then by \ref{PropCI14}, $(y \twoheadrightarrow z) \rightarrow (x \twoheadrightarrow z) = \top$, thus $(y \twoheadrightarrow z) \preceq (x \twoheadrightarrow z)$.
	
\end{proof}

%Note that \ref{CIMP} is a weak form of \ref{PropC2}.

% \proof{
% 	$(S12)$ Suppose $x\prec y$ and $y\prec z$, then by (S10) $x\prec y$ and $y\preceq z$. By (S6) $x\prec z$.
% }

We conclude from \ref{PropCI09} and \ref{PropCI10} that the relation ``$\ll$'' is transitive and antisymmetric, but it is not necessarily reflexive, whereas the relation ``$\preceq$'' is antisymmetric and reflexive by \ref{PropCI6} and \ref{PropCI1} respectively, but it is not necessarily transitive. The following proposition provides a condition for them to be partial orders.

\begin{proposition}
	The relation ``$\ll$" coincides with ``$\preceq$" if and only if ``$\ll$" is reflexive.	  
\end{proposition}

\begin{proof}
	Suppose $\ll$ is reflexive, then $x\preceq y$ implies $x\ll x\preceq y$, by \ref{PropCI4} $x\ll y$. The rest of the proof is trivial.
	
\end{proof}

%  \proof{($\Rightarrow$) If $x\preceq y$, then $x\prec x\preceq y$ and by (S6) $x\prec y$. By (S10), $\prec$ and $\preceq$ coincide. ($\Rightarrow$) trivial.}
%  
%  \begin{corollary}
%	  The reflexive closure of $\prec$ is $\preceq$.
%  \end{corollary}
%  
%Note that if $\langle A, \twoheadrightarrow,\rightarrow, \top\rangle$ is an SBCI algebra and $\ll$ is reflexive, then $\preceq$ is a partial order.

\begin{example}\label{ExSBCI2}
	The structure: $A=\prt{[0,1],\twoheadrightarrow_{GD},\rightarrow_{FD},1}$ is a SBCI, in which: 
	
	\begin{equation*}
	x\twoheadrightarrow_{GD} y=
	\begin{cases}
	1, \mbox{if $x\leq y$}\\
	y, \mbox{if $x > y$} 
	\end{cases}
	\end{equation*}
	
	and
	
	\begin{equation*}
	x \rightarrow_{FD} y=
	\begin{cases}
	1, \mbox{if $x\leq y$}\\
	max(1-x, y), \mbox{if $x > y$} 
	\end{cases}
	\end{equation*}
	and the following relations coincide: (1) $x\ll y$ if and only if $x\twoheadrightarrow_{GD} y=1$ if and only if $x \leq y$ and (2) $x\preceq y$ if and only if 
	$x\rightarrow_{FD} y= 1$ if and only if $x\leq y$.  In fact, ``$\ll$'' is reflexive and according to Baczy{\'n}ski et al \cite{Bac2008Book}[p.10, Table 1.4], $\twoheadrightarrow_{GD}$ satisfies \ref{Curry2-C} and and $\rightarrow_{FD}$ satisfies \ref{Curry2b-C}. \ref{CIScndPIsot} is also satisfied, for if $x\leq y$, then $x\twoheadrightarrow_{GD} y =1$ and therefore, if (i) $ z \leq x \leq y$, $z \twoheadrightarrow_{GD} x = z \twoheadrightarrow_{GD} y = 1$; (ii) $ x < z \leq y$, $z \twoheadrightarrow_{GD} x = x$ and $ z \twoheadrightarrow_{GD} y = 1$; (iii) $ x \leq y < z$, $z \twoheadrightarrow_{GD} x = x$ and $ z \twoheadrightarrow_{GD} y = y$. Now, if $x > y$, then $x\twoheadrightarrow_{GD} y = y$ and therefore, if (i) $ z \leq y < x$, $z \twoheadrightarrow_{GD} x = z \twoheadrightarrow_{GD} y = 1$; (ii) $ y < z \leq x$, $z \twoheadrightarrow_{GD} x = 1$ and $ z \twoheadrightarrow_{GD} y = y$; (iii) $ y < x \leq z$, $z \twoheadrightarrow_{GD} x = x$ and $ z \twoheadrightarrow_{GD} y = y$. In any case,  $x \twoheadrightarrow_{GD} y \leq (z \twoheadrightarrow_{GD} x) \rightarrow_{FD} (z \twoheadrightarrow_{GD} y)$.
	The remaining of SBCI-algebra properties are straightforward.
\end{example}

\begin{proposition} \label{prop-BCI-BCI-SBCI}
	Let $\langle A, \rightarrow_1,\top\rangle$ and $\langle A, \rightarrow_2,\top\rangle$ be  BCI-algebras with $\leq_1$ and $\leq_2$ as their 
	respective partial orders. If $\leq_2$ extends $\leq_1$ and $x\rightarrow_1 y\leq_2 x \rightarrow_2 y$ and   \begin{equation}\label{eq-ll-precsim-ll}
	w\leq_2 x\leq_1 y\leq_2 z\Longrightarrow w\leq_1 z
	\end{equation}
	for each $w,x,y,z\in A$, then         $\langle A,\rightarrow_1,\rightarrow_2,\top\rangle$ is an SBCI-algebra. The same applies to BCK-algebras and SBCK-algebras.
	
\end{proposition}

\begin{proof}      
	
	\begin{enumerate}[labelindent=\parindent, leftmargin=*,label=\normalfont{(SBCI\arabic*)}]
		\item Straightforward from \ref{Curry-C}.
		\item Straightforward from \ref{Curry-C}.
		
		\item By ($\mathcal{C}$-1), $(z\rightarrow_1  x)\rightarrow_1 
		((x\rightarrow_1 y)\rightarrow_1 (z\rightarrow_1 y))=\top$. So, by (A-5),  $(x\rightarrow_1 y)\rightarrow_1
		((z\rightarrow_1  x)\rightarrow_1 (z\rightarrow_1 y))=\top$, i.e. $(x\rightarrow_1 y)\leq_1
		((z\rightarrow_1  x)\rightarrow_1 (z\rightarrow_1 y))$. Thus, since $\leq_2$ extends $\leq_1$ and 
		$\rightarrow_1\leq_2 \rightarrow_2$, then  $(x\rightarrow_1 y)\leq_2
		((z\rightarrow_1  x)\rightarrow_2 (z\rightarrow_1 y))$.
		
		\item Straightforward from \ref{BCI-Pr8}.
		
		\item Straightforward  from (\ref{eq-ll-precsim-ll}) by taking $w=x$.
		
		\item Straightforward  from (\ref{eq-ll-precsim-ll}) by taking $y=z$.

		\item Straightforward from (A-4).
	\end{enumerate}

\end{proof}

\begin{corollary}
	If $\langle A, \rightarrow,\top\rangle$ is BCI-algebra, then  $\langle A,\rightarrow, \rightarrow,\top\rangle$ is an SBCI-algebra. The same applies to BCK-algebras and SBCK-algebra.
\end{corollary}

\begin{proof} Straightforward from Proposition \ref{prop-BCI-BCI-SBCI} once that (\ref{eq-ll-precsim-ll}) holds from \ref{BCI-Trans}.
	
\end{proof}

\begin{proposition}\label{PropSBCIBCI}
	Given an SBCI-algebra of the form $\langle A,\rightarrow, \rightarrow,\top\rangle$, the structure  $\langle A, \rightarrow,\top\rangle$ is a  
	BCI-algebra.
\end{proposition}
\begin{proof}
	\begin{enumerate}[labelindent=\parindent, leftmargin=*,label=\normalfont{(C-\arabic*)}]
		\item By \ref{CIScndPIsot},  $z\rightarrow x\preceq (y\rightarrow z)\rightarrow  (y\rightarrow x)$, i.e.
		$(z\rightarrow x)\rightarrow ( (y\rightarrow z)\rightarrow  (y\rightarrow x))=\top$. Therefore, by \ref{Curry2-C},  
		$(y\rightarrow z)\rightarrow ( (z\rightarrow x)\rightarrow  (y\rightarrow x))=\top$.
		
		\item By \ref{Curry2-C} and \ref{PropCI1},   $x\rightarrow ((x\rightarrow y)\rightarrow y)=(x\rightarrow y)\rightarrow(x\rightarrow y)=\top$.
		
		\item By \ref{PropCI1}.
		
		\item By \ref{PropCI6}.
	\end{enumerate}
	
\end{proof}

\begin{proposition}\label{SBCIReductisNotBCI}
	There are SBCI-algebras $\langle A,\twoheadrightarrow, \rightarrow,\top\rangle$  s.t.  the reduct $\langle A, \twoheadrightarrow,\top\rangle$ are not BCI-algebras.
\end{proposition}

\begin{proof}
	Consider the SBCI $A=\prt{[0,1],\twoheadrightarrow_R,\rightarrow_{LK}, 1}$ seen in Example \ref{ExSBCI}. The reduct $ \langle [0,1] , \twoheadrightarrow_{R}, 1 \rangle $ is not a BCI-algebra, since  	 
	$ x \twoheadrightarrow_{R} x = 1 - x + x^{2} \neq 1$, for all $x \in (0,1)$. Therefore, $(C-3)$ does not hold.
	
\end{proof}

%    \vspace{1em}
%      \textbf{Open Problem:} Is there a SBCI algebra, $\langle A,\twoheadrightarrow, \rightarrow,\top\rangle$,  s.t. the reduct $\langle A, \rightarrow,\top\rangle$ is not  
%      a BCI algebra?
%  
%  \vspace{1em}

\begin{proposition}
	There are algebras $\prt{A,\rightarrow,\top}$ in which the exchange principle \normalfont{\textbf{(EP)}} is satisfied, but 
	$\prt{A,\rightarrow,\rightarrow,\top}$ is not an SBCI-algebra.
\end{proposition}

\begin{proof}
	Take the  algebra $\prt{[0,1],\rightarrow_{\mbox{\tiny YG}},1}$, such that:
	
	\begin{equation*}
	x \rightarrow_{\mbox{\tiny YG}} y = 
	\begin{cases}
	1 \text{, if $x=y=0$}
	\\
	y^x \text{, otherwise}
	\end{cases}.
	\end{equation*}
	
	The implication ``$\rightarrow_{\mbox{\tiny YG}}$" satisfies \textbf{(EP)} --- see \cite{Bac2008Book}[p.10, Table 1.4]. However, it is easy to 
	check that it fails to satisfy \ref{PropCI1} --- take $x \in (0,1)$. 
	
	Alternatively, consider the algebra $\prt{[0,1],\rightarrow_{\mbox{\tiny WB}},1}$, such that:
	
	\begin{equation*}
	x \rightarrow_{\mbox{\tiny WB}} y = 
	\begin{cases}
	1 \text{, if $x<1$}
	\\
	y \text{, if $x=1$}
	\end{cases}.
	\end{equation*}
	
	This implication satisfies \ref{Curry2-C} \cite{Bac2008Book}[p.10, Table 1.4]. However, it is easy to check that it fails to satisfy \ref{PropCI6} --- take $x=0.5$ and $y=0.3$.
	
\end{proof}

%\begin{theorem}
% Given an SBCI algebra $A$ and $x,y,z\in A$, the following properties are satisfied: (\rot{Por que não leva este item para a proposição 4.1?})

% \begin{enumerate}[labelindent=\parindent, leftmargin=*,label=\normalfont{(SBCI\arabic*)}]
%\setcounter{enumi}{13}
%\item $x\twoheadrightarrow ((x\twoheadrightarrow y)\twoheadrightarrow y)=\top$, whenever $x\ll y$ \label{CIMP}
%	 	\item $x\rightarrow y=\top$ and $y\rightarrow x=\top$, then $x=y$. \label{PropCI9}.
%\end{enumerate}
%\end{theorem}

%\begin{proof}
% Suppose $x\ll y$, then  
% $x\twoheadrightarrow ((x\twoheadrightarrow y)\twoheadrightarrow y)\stackrel{WOP}{=} 
%x\twoheadrightarrow (\top\twoheadrightarrow y)\stackrel{\ref{CILeftNeut}}{=}x\twoheadrightarrow y\stackrel{WOP}{=}\top$.
%	  \ref{PropCI9} follows from  Proposition (\ref{prop-CI7-CI8}. \ref{PropCI10}).
% \end{proof}

% Note that \ref{CIMP} is a weak form of \ref{PropC2}.

\section{Relating Semi-BCI Algebras and Pseudo-BCI algebras}\label{SecPseudo}

The generalization of BCI/BCK-algebras is not new. In fact, G. Georgescu and A. Iorgulescu \cite{Georgescu2001} proposed an extension for BCK-algebras and later W. A. Dudek and Y. B. Jun \cite{Dudek} proposed an extension for BCI-algebras. The first  was called \textbf{Pseudo-BCK algebras } and the second \textbf{Pseudo-BCI algebras}. Like SBCI-algebras they propose two operations which generalize the primitive operation of BCK/BCI-algebras. Therefore a natural question arises: Are SBCI-algebras just a rewritten of those algebras? This section shows that the answer to this question is negative, meaning that SBCI/SBCK-algebras are completely new structures which generalize BCI-algebras. This section also shows how both structures are related.

\begin{definition}[\cite{Dymek2013}]
	A \textbf{pseudo-BCI algebra (PBCIA)} is a structure $\prt{A,\leq,\rightarrow,\rightsquigarrow,\top}$ s.t. ``$\leq$'' is a binary relation on the set $A$, ``$\rightarrow$'' and ``$\rightsquigarrow$'' are binary operations on $A$, $\top\in A$ and  for all $x,y,z\in A$:
	
	\begin{enumerate}[labelindent=\parindent, leftmargin=*,label=\normalfont{(PB-\arabic*)}]
		\item $x\rightarrow y\leq (y\rightarrow z)\rightsquigarrow(x\rightarrow z)$,
		\item  $x\rightsquigarrow y\leq (y\rightsquigarrow z)\rightarrow(x\rightsquigarrow z)$,\label{PB-2}
		\item $x\leq (x\rightarrow y)\rightsquigarrow y$,
		\item $x\leq (x\rightsquigarrow y)\rightarrow y$,\label{PB-4}
		\item\label{PB5}  $x\leq x$,
		\item if $x\leq y$ and $y\leq x$, then $x=y$,
		\item\label{PB7} $x\leq y\Leftrightarrow x\rightarrow y=\top\Leftrightarrow x\rightsquigarrow y=\top$.
	\end{enumerate}	 
\end{definition}

\begin{example} \label{ex.PBCI}
	The structure $\mathcal{A}= \langle \mathbb{R}^2, \preceq, \twoheadrightarrow, \rightarrow, (0, 0) \rangle$, such that $(x_1, y_1) \twoheadrightarrow (x_2, y_2) = (x_2 - x_1, (y_2 - y_1)e^{-x_1})$ and $(x_1, y_1) \rightarrow (x_2, y_2) = (x_2 - x_1, y_2 - y_1e^{x_2-x_1})$, is a PBCI-algebra.
\end{example}

The next proposition shows that PBCI-algebras are not suitable to model the intervalization of BCIs.

\begin{proposition}
	Let $\prt{A,\rightarrow,\top}$ be a BCI-algebra, $\prt{A,\preceq}$ be a meet semilattice, such that for each $x,y,z\in A$, $x\rightarrow (y\wedge z)= x \rightarrow y \wedge x \rightarrow z$, then the structure: $\prt{\mathbb{A},\Twoheadrightarrow,\Rightarrow,[\top,\top]}$, where $\Twoheadrightarrow$ and $\Rightarrow$ are defined in Theorem \ref{ThIntBCI}, is not a PBCI-algebra.
\end{proposition}

\begin{proof}
	By \ref{PB5} and \ref{PB7}, for all $X \in \mathbb{A}$, $X \Twoheadrightarrow X = [\top,\top]$ should hold, however by Corollary \ref{CorolXiXDeg}, this only applies if $X$ is degenerate.
	
\end{proof}

Given a PBCI-algebra, if the relations: (a) $x\leq_1 y\Leftrightarrow x\rightarrow y=\top$ and (b) $x\leq_2 y \Leftrightarrow x\rightsquigarrow y=\top$ are defined, then the axiom  \ref{PB7} imposes that they  must coincide. In the case of SBCI-algebras the relations $\ll$ and $\preceq$ does not necessarily coincide. Moreover, even if they coincide there are SBCIs which are not PBCIs --- see Proposition \ref{PropPseudoSBCICoin}. Finally,  there are SBCIs in which the relation ``$\ll$'' can be irreflexive  refuting  the axiom \ref{PB5} (see Proposition \ref{SBCIReductisNotBCI}). Therefore, this lead us to conclude that SBCI and PBCI-algebras are different structures. 

\begin{proposition} \label{PropPseudoSBCICoin}
	There are SBCI-algebras $\langle A,\twoheadrightarrow, \rightarrow,\top\rangle$ whose relations ``$\ll$" and ``$\preceq$" coincide but are not PBCI-algebras.
\end{proposition}

\begin{proof}
	The SBCI $\mathcal{A} = \langle[0,1],\twoheadrightarrow_{GD},\rightarrow_{FD},1\rangle$ provided at Example \ref{ExSBCI2} is not a PBCI-algebra. In fact, take $x=\dfrac{3}{4}$, $y =\dfrac{1}{2}$ and $z = \dfrac{1}{5}$, then $x \rightarrow_{FD} y = \dfrac{1}{2} $ and $(y \rightarrow_{FD} z) \twoheadrightarrow_{GD} (x \rightarrow_{FD} z) = \dfrac{1}{4} $, but since the relations ``$\ll$" and ``$\preceq$" coincide with the usual order,  \ref{PB-2} is not satisfied.
	
\end{proof}

\begin{proposition} \label{PBCI.NaoSBCI}
	There are PBCI-algebras $\langle A, \preceq, \twoheadrightarrow, \rightarrow,\top\rangle$ which are not SBCI-algebras.
\end{proposition}

\begin{proof}
	The PBCI $\mathcal{A}= \langle \mathbb{R}^2, \preceq, \twoheadrightarrow, \rightarrow, (0, 0) \rangle$ presented in Example \ref{ex.PBCI} is not a SBCI-algebra. In fact, take $(x_1, y_1), (x_2, y_2) \in \mathbb{R}^2$ such that $y_2(e^{x_1}-1) \neq y_1(e^{x_2}-1)$ and any $(x_3, y_3) \in \mathbb{R}^2 $. Note that:
	\begin{eqnarray*}
		(x_1, &y_1&) \twoheadrightarrow ((x_2, y_2) \twoheadrightarrow (x_3, y_3)) = \\
		&=& (x_1, y_1) \twoheadrightarrow (x_3 -x_2, y_3 - y_2e^{x_3 - x_2}) \\
		&=& ((x_3 -x_2) - x_1, (y_3 - y_2e^{x_3 - x_2}) - y_1e^{(x_3 -x_2) - x_1}) \\
		&=& ( x_3 -x_2 - x_1, y_3 - y_2e^{x_3 - x_2} - y_1e^{x_3 -x_2 - x_1})
	\end{eqnarray*}
	and
	\begin{eqnarray*}
		(x_2, &y_2&) \twoheadrightarrow ((x_1, y_1) \twoheadrightarrow (x_3, y_3)) = \\
		&=& (x_2, y_2) \twoheadrightarrow (x_3 -x_1, y_3 - y_1e^{x_3 - x_1}) \\
		&=& ((x_3 -x_1) - x_2, (y_3 - y_1e^{x_3 - x_1}) - y_2e^{(x_3 -x_1) - x_2}) \\
		&=& ( x_3 -x_2 - x_1, y_3 - y_1e^{x_3 - x_1} - y_2e^{x_3 -x_2 - x_1}).
	\end{eqnarray*}
	Since $y_2(e^{x_1}-1) \neq y_1(e^{x_2}-1)$, then $ y_3 - y_2e^{x_3 - x_2} - y_1e^{x_3 -x_2 - x_1} \neq y_3 - y_1e^{x_3 - x_1} - y_2e^{x_3 -x_2 - x_1}$, and therefore $\mathcal{A}$ does not satisfy the axiom \ref{Curry2-C}.
	
\end{proof}

The next proposition ensures that the intersection between the PBCIs and SBCIs is formed only by  BCI-algebras.

\begin{proposition} \label{setasiguais}
	Let $ \mathcal{A} = \langle A, \twoheadrightarrow, \rightarrow,\top \rangle$ be a SBCI-algebra, s.t. the relations: ``$\ll$" and ``$\preceq$" correspond to the operations: “$\twoheadrightarrow$” and “$\rightarrow$”, respectively. If $ \mathcal{A}$ is also a PBCI-algebra, then $ \langle A, \rightarrow,\top \rangle $ is a BCI-algebra.
\end{proposition}

\begin{proof}
	Since $ \mathcal{A}$ is also a PBCI-algebra then the relations ``$\preceq$" and ``$\ll$" coincide. Now, by \ref{PB-4},
	\begin{eqnarray*}
		x \preceq (x \rightarrow y) \twoheadrightarrow y &\Rightarrow& x \ll (x \rightarrow y) \twoheadrightarrow y\\
		&\Rightarrow& x \twoheadrightarrow ((x \rightarrow y) \twoheadrightarrow y) = \top \\
		&\stackrel{\ref{Curry2-C}}{\Rightarrow}& (x \rightarrow y)\twoheadrightarrow (x \twoheadrightarrow y) = \top \\
		&\Rightarrow& (x \rightarrow y) \ll (x \twoheadrightarrow y)\\	
		&\Rightarrow& (x \rightarrow y) \preceq (x \twoheadrightarrow y).			
	\end{eqnarray*}
	Therefore, since $x \twoheadrightarrow y \preceq x \rightarrow y$ \ref{PropCI15} for all $x, y \in A$ then, from axiom \ref{PropCI6}, follows $x \twoheadrightarrow y = x \rightarrow y$, for all $x, y \in A$. We conclude by Proposition \ref{PropSBCIBCI} that $ \langle A, \rightarrow,\top \rangle $ is a BCI-algebra.	
	
\end{proof}

Figure 1 shows how PBCI and SBCI algebras are classified.

\begin{center}
	\begin{figure}[h]
		\centering
		\includegraphics[scale=0.15]{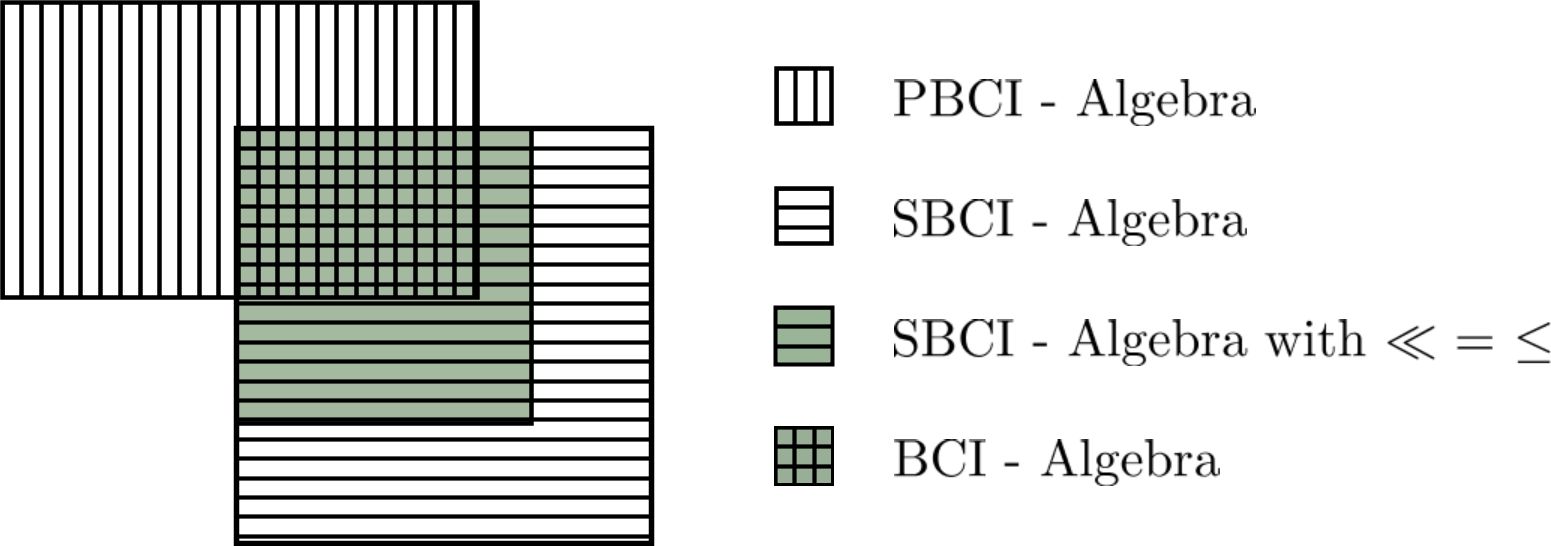}
		\caption{Relation between SBCI and PBCI-algebras.}
	\end{figure}
\end{center}

\section{Final Remarks}\label{SecFinalRem}

	This paper proposes a new algebraic structure which generalizes the notion of BCI-algebra. It is an algebraic structure which captures the most important properties of a Fuzzy Implication after it has been ``intervalized'' in a correct way. The resulting structures, called \textbf{Semi-BCI algebras}, capture the properties of  the structures which arise from the ``intervalization'' of BCI-algebras. In other words, as BCI-algebras abstract the \L ukasiewicz algebra, the SBCI-algebras abstract the respective intervalization of the \L ukasiewicz algebra. The paper also provides the connection of such structures with PBCI-algebras.
	
	As further steps, the authors aim  to investigate more closely this structure.
	Some questions like the logical counter-part of such algebras requires a deeper investigation. Entities like filters, ideals, category and others require investigation. In other words, the generalization of BCI-algebras in terms of its intervalization is provided from algebraic viewpoint, but the logical correspondence (The resulting Interval Fuzzy Logic) requires more
	reflection. Contrary to BL-algebras (BCIs with extra properties) a question is posed:  
	
	\begin{quotation}
		``Is interval correctness incompatible with logical principles, like the notion of deducibility and its whole connection with implications?''
	\end{quotation}
	
	This question is important, since the answer can limit the term: ``Interval Fuzzy Logic'' in a broad sense.
	
\section*{Acknowledgments}  This work was  supported  by the Brazilian funding agency  \textbf{CNPq (Brazilian Research Council)} under Projects: \textsf{304597/2015-5, 311796/2015-0, 482809/2013-2 and 307781/2016-0}), \textbf{Marie Curie (EU-FP7)} under the project: PIRSES-GA-2012-318986 and \textbf{R\&D Unit 50008}, financed by the applicable financial framework (FCT/MEC through national funds and when applicable co-funded by \textbf{FEDER – PT2020} partnership agreement).

\end{document}